\documentclass[10pt,letterpaper,twocolumn]{article}

\usepackage{arxiv}
\usepackage{orcidlink}

\titlespacing*{\section}{0pt}{1.2ex plus 0.3ex minus 0.2ex}{0.6ex}
\hypersetup{hidelinks}
\newcommand{\PaperBalanceLastPage}{}
\newcommand{\PaperIntroStart}{Credential}
\newcommand{\PaperRunIn}[1]{\paragraph{#1}}

\newcommand{\PaperFrontMatter}{%
  \hypersetup{
    pdftitle={\PaperTitle},
    pdfauthor={\PaperAuthorsMetadata},
    pdfsubject={Preprint manuscript},
    pdfkeywords={\PaperKeywords}
  }
  \twocolumn[
    \begin{center}
      \vspace*{-1.2em}
      {\sffamily\mdseries\fontsize{24}{27}\selectfont \PaperTitle\par}
      \vspace{0.8em}
      {\sffamily\normalsize \PaperAuthorOne{}\,\orcidlink{\PaperAuthorOneORCID} and \PaperAuthorTwo{}\,\orcidlink{\PaperAuthorTwoORCID}\par}
      \vspace{0.85em}
    \end{center}
    \begin{center}
      \begin{minipage}{0.91\textwidth}
        \sffamily\small
        \noindent\textbf{Abstract\textemdash}\PaperAbstract
        \par\smallskip
        \noindent\textbf{Index Terms\textemdash}\PaperKeywords
      \end{minipage}
    \end{center}
    \vspace{0.5em}
  ]
  \begingroup
    \renewcommand{\thefootnote}{}%
    \footnotetext{%
      \begin{list}{\textbullet}{%
        \setlength{\leftmargin}{1.2em}%
        \setlength{\labelwidth}{0.7em}%
        \setlength{\labelsep}{0.3em}%
        \setlength{\itemindent}{0pt}%
        \setlength{\listparindent}{0pt}%
        \setlength{\topsep}{0pt}%
        \setlength{\parsep}{0pt}%
        \setlength{\itemsep}{0pt}%
      }
        \item Corresponding author: \PaperCorrespondingAuthor{} (\href{mailto:\PaperCorrespondingEmail}{\PaperCorrespondingEmail}).
        \item \PaperAuthorOne{} and \PaperAuthorTwo{} are with \PaperAffiliationAddress.
      \end{list}%
    }%
  \endgroup
}

\begin{document}
\providecommand{\TrainTicketServices}{32}
\providecommand{\TrainTicketEdges}{72}
\providecommand{\TrainTicketTraces}{1157}
\providecommand{\FormalImpactChecks}{132}
\providecommand{\FormalGridCalls}{11}
\providecommand{\FormalSweepCallsLow}{4.7}
\providecommand{\FormalSweepCallsHigh}{5.0}
\providecommand{\BudgetPathCases}{60}
\providecommand{\SolverQualityCases}{60}
\providecommand{\ReplayWorkloads}{6}
\providecommand{\ReplayMaxWidth}{6}
\providecommand{\ReplayHeterogeneousCases}{36}
\providecommand{\ReplayTrainAlphaDPRatio}{0.09}
\providecommand{\ChainLocalMaxGap}{84.30}
\providecommand{\ChainLocalSameAssignment}{43}
\providecommand{\JointChainStagedWorseCases}{117}
\providecommand{\JointChainStagedWorseTotal}{120}
\providecommand{\JointDepthStagedWorseCases}{78}
\providecommand{\JointDepthStagedWorseTotal}{110}
\providecommand{\JointStagedWorseCases}{195}
\providecommand{\JointStagedWorseTotal}{230}
\providecommand{\TrainHRatio}{0.685}
\providecommand{\TrainSRatio}{0.743}
\providecommand{\TrainDomains}{6}
\providecommand{\TrainCountMatch}{80}
\providecommand{\TrainSolverMatch}{80}
\providecommand{\GuardrailDisagreements}{6}
\providecommand{\GuardrailCandidateRecoveries}{5}
\providecommand{\OverlayCompiledRecoveries}{11}
\providecommand{\OverlayComparisons}{12}
\providecommand{\OverlayMaxGap}{5.86}
\providecommand{\ScalingRuntimeFifty}{1.95}
\providecommand{\ScalingRuntimeThousand}{17.03}
\providecommand{\ScalingRuntimeFactor}{8.7}
\providecommand{\RuntimeMachine}{spark (aarch64)}
\providecommand{\RuntimePython}{3.13.14}
\providecommand{\RuntimeNumPy}{2.5.2}
\providecommand{\CriticalStableRows}{4}
\providecommand{\CriticalHUniform}{0.685}
\providecommand{\CriticalHWeighted}{0.711}
\providecommand{\CriticalSUniform}{0.743}
\providecommand{\CriticalSWeighted}{0.780}
\providecommand{\CriticalLooseMaxChange}{0.011}
\providecommand{\BRExactMaxRegret}{1.22}
\providecommand{\SmallGapHeuristicSlack}{0.018}
\providecommand{\SmallGapExactSlack}{0.002}
\providecommand{\SmallGapNearTieBudget}{0.60}
\providecommand{\PQCMeasurementObservations}{12000}
\providecommand{\PQCMeasurementBlocks}{4}
\providecommand{\PQCProfileCount}{3}
\providecommand{\PQCNetworkCount}{5}
\providecommand{\PQCAOneCostLow}{0.0338}
\providecommand{\PQCAOneCostHigh}{15.5815}
\providecommand{\PQCChecksPerPath}{10}
\providecommand{\PQCNTwoEdgePasses}{8}
\providecommand{\PQCNZeroGraphPasses}{7}
\providecommand{\PQCHeadroomChecksPerRule}{20}
\providecommand{\PQCMethodMaxReduction}{36}
\providecommand{\PQCNZeroRiskExactTight}{0.693}
\providecommand{\PQCNZeroRiskExactLoose}{0.638}
\providecommand{\PQCNZeroStagedExact}{0.884}

\newcommand{\PaperTitle}{Optimizing Credential Blast Radius Through Trust Boundaries and Delegation Under Post-Quantum Authentication Costs}
\newcommand{\PaperAuthorOne}{Pauli Taipale}
\newcommand{\PaperAuthorTwo}{Harri Lainio}
\newcommand{\PaperAuthorOneORCID}{0009-0000-0280-3243}
\newcommand{\PaperAuthorTwoORCID}{0009-0009-8562-3955}
\newcommand{\PaperAuthorsMetadata}{Pauli Taipale, Harri Lainio}
\newcommand{\PaperAffiliation}{OP Lab, OP Pohjola, Helsinki, Finland}
\newcommand{\PaperAffiliationAddress}{OP Lab, OP Pohjola, Gebhardinaukio 1, FI-00510 Helsinki, Finland}
\newcommand{\PaperCorrespondingAuthor}{Pauli Taipale}
\newcommand{\PaperCorrespondingEmail}{pauli.taipale@op.fi}
\newcommand{\PaperKeywords}{access control, authentication, graph theory, optimization, public key}

\newcommand{\PaperAbstract}{%
Partitioning interacting services into independently rooted trust domains limits issuer-compromise reach while increasing calls across trust boundaries.
Post-quantum replacements for public-key authentication and key-establishment mechanisms can increase crossing latency on constrained or lossy paths.
We formulate the joint selection of trust domains and credential-derivation structures under policy and latency constraints, linking separate service-interaction and credential-derivation graphs through domain assignment.
Credential blast radius measures weighted service impact after compromise. A linear upper bound supports optimization, while a joint event model gives exact expected impact.
We identify when risk from issuers trusted across domains can be incorporated into this linear score, avoiding separate issuer-propagation calculations for each candidate.
Although the general problem is NP-hard, we identify restricted cases that can be solved efficiently and exactly.
Joint optimization yields lower blast radius than choosing boundaries first in \JointStagedWorseCases{} of \JointStagedWorseTotal{} exhaustive synthetic comparisons, especially under chained delegation.
A trace-derived replay used measured post-quantum costs, synthetic risk inputs, a fixed derivation family, and one to six trust domains. Under independent compromise events, mean expected impact was up to \PQCMethodMaxReduction\% lower than with one domain within the latency budget.
The framework turns risk assumptions and measured crossing costs into candidate trust-domain and credential-derivation designs.
}

\PaperFrontMatter

\section{Introduction}
\PaperIntroStart{} compromise can propagate through two distinct mechanisms: compromise of a domain root exposes every principal in its domain, while compromise of a delegated principal exposes the descendants reachable through credential derivation.
Trust boundaries limit the first mechanism. Derivation structure limits the second.
Zero-trust guidance describes trust and identity controls architecturally.\cite{NIST800207}
We study how to choose trust boundaries and credential-derivation structures jointly when boundary crossings carry a performance cost.

Post-quantum cryptography (PQC) makes this joint decision consequential.
Classical public-key key establishment and signatures are vulnerable to Shor's algorithm, whereas Grover's generic search gives a quadratic quantum speedup against symmetric keys.\cite{Shor1994,Grover1996}
Thus, a 256-bit symmetric key retains approximately 128 bits of key-search strength in the idealized query model. Classical RSA- and elliptic-curve-based boundary mechanisms require replacement.
Post-quantum migration changes key establishment and authentication separately.
ML-KEM can replace or augment classical key establishment, while ML-DSA and SLH-DSA replace signatures used for peer authentication and certificate chains.
On some platforms, implementations of post-quantum key-establishment and signature schemes can match or outperform their classical counterparts, so PQC does not impose a uniform computational slowdown.
Larger public keys, ciphertexts, signatures, and certificates can nevertheless increase handshake byte volume and the number of transport packets or segments, amplifying delay on constrained or lossy paths.\cite{Sosnowski2023PQTLS13,KampanakisChildsKlein2024,Sim2025KpqC}
A boundary that improves containment can therefore consume a measurable latency or bandwidth budget whose dominant source depends on the deployment.
The U.S. National Institute of Standards and Technology (NIST) has standardized several PQC algorithms, while transition guidance gives a timeline for retiring quantum-vulnerable public-key algorithms in NIST standards.\cite{NISTIR8547ipd,FIPS203,FIPS204,FIPS205}

High-rate systems already amortize classical public-key authentication.
Google's Application Layer Transport Security uses resumption and long-lived remote-procedure-call channels to preserve workload authentication while reducing repeated cryptographic work.\cite{Google2017ALTS}
At the constrained-device extreme, smart-card experiments found hybrid post-quantum payment transactions to be dominated by transmitting larger certificate chains over the card interface rather than by cryptographic computation.\cite{BettaleDeOliveiraDottax2022}
Certificate-size reduction efforts for post-quantum HTTPS address the same communication bottleneck.\cite{Google2026MTC}
These cases motivate deciding where independently rooted authentication boundaries justify their path-dependent cost.
Figure~\ref{fig:overview} summarizes the two architecture decisions and the quantities used to evaluate them.

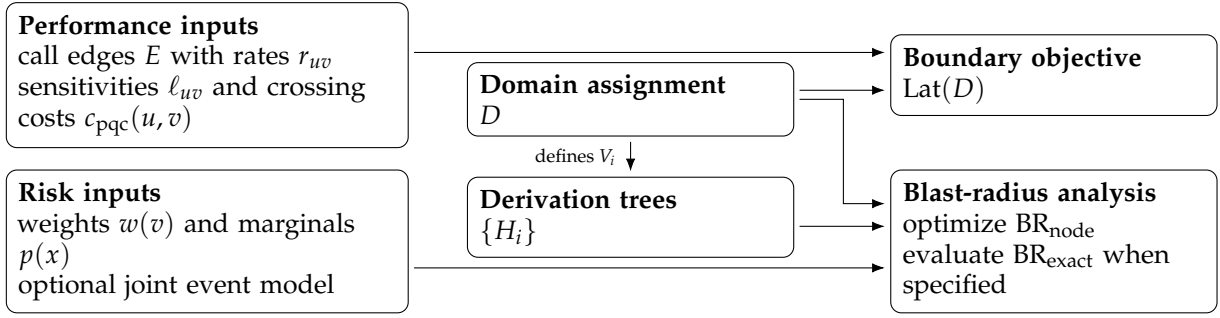
\begin{figure*}[tbp]
  \centering
  \begin{tikzpicture}[>=Latex,x=1cm,y=1cm]
  \node[draw,rounded corners,align=left,inner sep=4.5pt,text width=5.0cm] (inputs1) at (0.9,1.1)
    {\raggedright\textbf{Performance inputs}\\ call edges $E$ with rates $r_{uv}$\\ sensitivities $\ell_{uv}$ and crossing costs $c_{\mathrm{pqc}}(u,v)$};
  \node[draw,rounded corners,align=left,inner sep=4.5pt,text width=5.0cm] (inputs2) at (0.9,-1.1)
    {\raggedright\textbf{Risk inputs}\\ weights $w(v)$ and marginals $p(x)$\\ optional joint event model};

  \node[draw,rounded corners,align=left,inner sep=4.5pt,text width=4.0cm] (domains) at (6.5,0.78)
    {\raggedright\textbf{Domain assignment}\\ $D$};
  \node[draw,rounded corners,align=left,inner sep=4.5pt,text width=4.0cm] (derivations) at (6.5,-0.78)
    {\raggedright\textbf{Derivation trees}\\ $\{H_i\}$};

  \node[draw,rounded corners,align=left,inner sep=4.5pt,text width=4.0cm] (lat) at (12.1,1.1)
    {\raggedright\textbf{Boundary objective}\\ $\mathrm{Lat}(D)$};
  \node[draw,rounded corners,align=left,inner sep=4.5pt,text width=4.0cm] (br) at (12.1,-1.1)
    {\raggedright\textbf{Blast-radius analysis}\\ optimize $\mathrm{BR}_{\mathrm{node}}$\\ evaluate $\mathrm{BR}_{\mathrm{exact}}$ when specified};

  \tikzset{_couplingarrow/.style={->,shorten >=2pt,shorten <=2pt}}
  \draw[_couplingarrow] ([yshift=3mm]inputs1.east) -- ([yshift=3mm]lat.west);
  \draw[_couplingarrow] ([yshift=-3.5mm]inputs2.east) -- ([yshift=-3.5mm]br.west);
  \draw[_couplingarrow] (domains.south) -- node[midway,left=2pt,font=\scriptsize] {defines $V_i$} (derivations.north);
  \draw[_couplingarrow] ([yshift=1.2mm]domains.east) -- ([yshift=-2mm]lat.west);
  \draw[_couplingarrow] (domains.east) -- ++(0.65,0) |- ([yshift=5mm]br.west);
  \draw[_couplingarrow] ([yshift=-1.2mm]derivations.east) -- ([yshift=2mm]br.west);
\end{tikzpicture}
  \caption{\label{fig:overview}
  Two coupled architecture decisions and their evaluation.
  Performance inputs describe service calls and path-dependent crossing costs.
  Risk inputs describe service impact, marginal compromise probabilities, and optional event dependence.
  Domain assignment $D$ defines $V_i=D^{-1}(i)$, the services in domain $i$.
  Credential-derivation tree $H_i$ is built on $V_i$ and determines compromise reach within that domain.
  Rightward arrows show objective dependence.
  Performance inputs and $D$ determine boundary latency $\mathrm{Lat}(D)$.
  Risk inputs, $D$, and $\{H_i\}$ determine the conservative linear blast-radius score $\mathrm{BR}_{\mathrm{node}}$.
  Given a joint compromise model, $\mathrm{BR}_{\mathrm{exact}}$ gives exact expected impacted weight.
  Joint optimization trades boundary latency against compromise reach.}
\end{figure*}

We model operational principals in separate service-interaction and credential-derivation graphs.
The service-interaction graph assigns a calibrated crossing cost to calls that cross trust domains.
Within each domain, a constrained credential-derivation tree determines how compromise of a delegating principal reaches downstream principals.
The design chooses both under policy and latency constraints, including fanout and depth limits.
Domain-root compromise remains confined to its domain.

The tree model does not cover every credential system.
In shared-issuer JSON Web Token (JWT) or OpenID Connect (OIDC) deployments, services in several nominal domains may accept credentials from one issuer, creating compromise paths that bypass the domain roots.
We model these acceptance relationships separately. A structural condition identifies when issuer risk is represented exactly by the additive score. Otherwise, candidate designs must be evaluated by tracing issuer reachability explicitly.

This work contributes:
\begin{itemize}
\item A two-layer optimization model that chooses trust domains and credential-derivation trees under policy and latency constraints. The optimizer uses a conservative linear objective. When a joint compromise model is available, candidates are evaluated by exact expected impact.
\item Hardness and tractable regimes for the model, including a condition for exact additive scoring of shared issuers. The solver map distinguishes results proved for this formulation from guarantees inherited from established optimization algorithms.
\item An operational evaluation that measures named Transport Layer Security (TLS) profiles across network paths, applies the resulting costs to trace-derived optimization, and validates selected crossings on held-out measurements. It also tests whether additive shared-issuer scoring agrees with explicit issuer-reachability evaluation.
\end{itemize}
Given explicit workload, compromise, criticality, and credential-semantics assumptions, the planning method returns containment--performance tradeoffs and states which guarantees apply.
Proofs and supporting experiments appear in the supplementary material.

\section{Background and threat model}\label{sec:background}
A trust domain is a set of principals that derive session keys or credentials from a common issuer, signing key, or key-management root $\rho_i$.
We study key-establishment domains: after establishment, intra-domain payload interactions use symmetric protection or already-issued delegated credentials, whereas a crossing connects independently rooted credential systems and invokes public-key authentication or key establishment.
This is a deployment assumption rather than a universal definition of a trust domain.
It need not coincide with a subnet or administrative boundary: one network segment can contain several isolated issuers, while separate segments can still trust one shared issuer.
We therefore distinguish the service-interaction graph, the trust-domain assignment, and credential propagation.

We consider hybrid architectures that replace the quantum-vulnerable public-key mechanism at these crossings with PQC while retaining symmetric protection and derived credentials within domains.
A compromised domain root can mint for its entire domain.
A compromised service or delegator can impersonate the descendants enabled by its credential authority.
Boundary placement controls the former reach. Derivation design controls the latter.
Services holding identities with different scopes can be split into separate principals before optimization, as described in Supplementary Sec.~\ref{SIsup:deployment}.

Compromise probabilities $p(x)$ are marginal scenario inputs over a fixed planning horizon $T$, a chosen time interval such as a detection or rotation window. Service weights $w(v)$ encode impact.
Here compromise means credential theft or loss of control over credential authority, not quantum cryptanalysis. PQC enters through the cost of replacing quantum-vulnerable boundary mechanisms.
We assume that verifiers and boundary enforcement follow the modeled policy.
Verifier compromise, software supply-chain compromise, denial of service, and non-credential lateral movement are outside scope.
The model is conditional on its $(p,w)$ scenario and does not estimate compromise probabilities.

\section{Formal model}\label{sec:model}
Let $G=(V,E)$ be the directed service-interaction graph, where $V$ is the set of services or principals and $(u,v)\in E$ means that $u$ calls $v$.
Each edge has an interaction rate $r_{uv}\ge0$ and system-level latency sensitivity $\ell_{uv}\ge0$.
The interaction rate, latency sensitivity, and calibrated crossing cost together determine boundary-authentication overhead rather than payload-cryptography work.
Policy-forbidden calls are removed or assigned zero weight.
A design uses $k\ge1$ nonempty trust domains.
A domain assignment
\[
D:V\rightarrow\{1,\ldots,k\}
\]
induces $V_i=\{v:D(v)=i\}$. An interaction crosses a boundary when $D(u)\ne D(v)$.

Credential derivation in domain $i$ is
\begin{equation}
H_i=(V_i\cup\{\rho_i\},A_i),
\end{equation}
where virtual root $\rho_i$ represents the issuer or master secret and $A_i$ is the set of directed derivation arcs. The relation $(x,y)\in A_i$ means compromise of $x$ enables credentials for $y$.
The core model restricts $H_i$ to a rooted directed arborescence: every service has exactly one incoming derivation arc and is reachable from $\rho_i$.
This captures direct issuance, delegation chains, and shallow subissuer layouts. Because each service has one parent, every root-to-service compromise path is explicit.
The separate acceptance model below handles issuer entry points not represented by these roots.

\section{Problem definition}\label{sec:problem}
Let $\mathcal A^{\mathrm{allow}}\subseteq (V\cup\{\rho_1,\ldots,\rho_k\})\times V$ denote the permitted derivation arcs. When this relation is omitted, every root-to-service arc and every service-to-service arc between distinct vertices is eligible.
The inputs are $G$, $r$, $\ell$, service weights $w$, compromise probabilities $p$, domain count $k$, derivation limits $(\Delta,h)$, $\mathcal A^{\mathrm{allow}}$, optional domain-policy constraints, and the effective crossing-cost function $c_{\mathrm{pqc}}$ defined in Sec.~\ref{sec:latency}.
Service weights satisfy $w(v)\ge0$, compromise probabilities satisfy $p(x)\in[0,1]$, $\Delta$ is a positive integer, and $h$ is a nonnegative integer.
The decision variables are $D$ and $\{H_i\}$.
A design is feasible, written $(D,\{H_i\})\in\mathcal{F}$, when it satisfies:
\begin{description}
\item[(C1)] \textbf{Partition.} Every service belongs to exactly one nonempty domain.
\item[(C2)] \textbf{Policy.} Required must-link pairs share a domain and cannot-link pairs do not.
\item[(C3)] \textbf{Derivation.} Each $H_i$ is a rooted directed arborescence spanning $V_i$, with $A_i\subseteq\mathcal A^{\mathrm{allow}}$.
\item[(C4)] \textbf{Operational limits.} Every derivation out-degree is at most $\Delta$ and every root-to-leaf depth is at most $h$.
\end{description}
In particular, (C3)--(C4) imply
\begin{equation}
|V_i|\le
\begin{cases}
h, & \Delta=1,\\
\sum_{d=1}^{h}\Delta^d, & \Delta>1.
\end{cases}
\label{eq:capacity}
\end{equation}
This condition is necessary under (C3)--(C4). It is also sufficient when derivation eligibility is complete and no additional policy constraint applies.
Capacity and latency constrain $k$ jointly with $\Delta$ and $h$.
A heterogeneous root risk must be attached to a concrete issuer label, and policy determines which services may use that issuer.
Interchangeable roots use homogeneous risk, whose root contribution is assignment-independent (Supplementary Sec.~\ref{SIsup:risk_summaries}).

\section{Blast-radius objective}\label{sec:br_objectives}
For $S\subseteq V_i\cup\{\rho_i\}$, let $\mathrm{Reach}_{H_i}(S)$ be the services reachable from $S$ in $H_i$ and define deterministic impact as
\[
\mathrm{BR}_i(S,H_i)=\sum_{v\in\mathrm{Reach}_{H_i}(S)}w(v).
\]
Our primary objective sums each compromise point's probability-weighted impact:
\begin{equation}
\mathrm{BR}_{\mathrm{node}}(D,\{H_i\})
=
\sum_{i=1}^{k}\sum_{x\in V_i\cup\{\rho_i\}}
p(x)\,\mathrm{BR}_i(\{x\},H_i).
\label{eq:br_node}
\end{equation}
A service reachable from several compromise points contributes once for each point.
For an arborescence, let $\mathrm{Anc}_{H_i}(v)$ be the vertices on the root-to-$v$ path, including both $\rho_i$ and $v$.
Swapping the sums gives
\begin{equation}
\mathrm{BR}_{\mathrm{node}}
=
\sum_{i=1}^{k}\sum_{v\in V_i}w(v)
\sum_{x\in\mathrm{Anc}_{H_i}(v)}p(x).
\label{eq:br_ancestor}
\end{equation}
This form makes the design pressure explicit: high-probability principals should control little descendant weight.

Give each derivation arc entering a non-root vertex $x$ length $p(x)$, and let $d_{H_i}(\rho_i,v)$ be the resulting root-to-$v$ distance.
Then Eq.~(\ref{eq:br_ancestor}) is equivalently
\begin{equation}
\mathrm{BR}_{\mathrm{node}}
=
\sum_{i=1}^{k}\left[
p(\rho_i)W(V_i)+\sum_{v\in V_i}w(v)d_{H_i}(\rho_i,v)
\right],
\label{eq:br_root_distance}
\end{equation}
where $W(S)=\sum_{v\in S}w(v)$.
For a fixed service set $S$, let $\mathcal H_i(S)$ be the derivation trees on $S\cup\{\rho_i\}$ that satisfy (C3)--(C4), including $A_i\subseteq\mathcal A^{\mathrm{allow}}$, and the applicable domain-policy constraints.
Define $\Phi_i(S)$ as the minimum bracketed term over this family, with $\Phi_i(S)=+\infty$ when the family is empty.
For fixed $D$, each domain can attain $\Phi_i(V_i)$ independently. The outer assignment remains coupled because it determines both $V_i$ and the interaction cut.

Let $E_x$ denote compromise of $x$.
The exact expected impacted weight is
\begin{equation}
\mathrm{BR}_{\mathrm{exact}}
=
\sum_{i=1}^{k}\sum_{v\in V_i}w(v)
\Pr\!\left(\bigcup_{x\in\mathrm{Anc}_{H_i}(v)}E_x\right).
\label{eq:br_exact}
\end{equation}
The union bound gives $\mathrm{BR}_{\mathrm{exact}}\le\mathrm{BR}_{\mathrm{node}}$ under any dependence structure.
Under independence the probability is $1-\prod_x(1-p(x))$, so the bound is first-order tight when probabilities are small.
We optimize the conservative linear form because it avoids an assumed dependence model and preserves additive structure. When probabilities are large or strongly dependent, candidates should be rescored under an explicit joint model.

Worst-case single-compromise and issuer-only summaries are diagnostics rather than additional propagation models.
Alternative risk summaries appear in Supplementary Sec.~\ref{SIsup:risk_summaries}. A surrogate-regret check appears in Supplementary Table~\ref{SItab:br_exact_regret_stress}.

\section{Latency cost}\label{sec:latency}
For each interaction $(u,v)$, let $r_{uv}$ be its rate, $\ell_{uv}$ its system-level latency sensitivity, and $c_{\mathrm{pqc}}(u,v)$ the effective cost of invoking the boundary mechanism.
When latency is expressed in ms/request, $r_{uv}$ counts calls per top-level request, $\ell_{uv}$ is dimensionless, and $c_{\mathrm{pqc}}(u,v)$ is in milliseconds per crossing call.
For a predicate $P$, $\mathbf{1}[P]$ is $1$ when $P$ is true and $0$ otherwise.
The boundary objective is
\begin{equation}
  \mathrm{Lat}(D)
  =
  \sum_{(u,v)\in E} r_{uv}\ell_{uv}c_{\mathrm{pqc}}(u,v)
  \mathbf{1}[D(u)\ne D(v)].
  \label{eq:lat_simple}
\end{equation}
Equation~(\ref{eq:lat_simple}) is a linearized boundary-overhead proxy, not a general end-to-end or tail-latency model.
The formal results require only nonnegative edge coefficients and do not depend on PQC specifically. PQC supplies the motivating calibration regime.
Accordingly, $c_{\mathrm{pqc}}(u,v)$ is a deployment-specific coefficient.
It is a measured or estimated profile for a selected cryptographic suite, protocol stack, hardware platform, network path, and reuse policy, and different edges may use different profiles.
The product $r_{uv}c_{\mathrm{pqc}}(u,v)$ represents key-establishment or authentication work after session reuse, credential caching, and key-update policy.
Represent reuse through either an event-equivalent rate or an amortized per-interaction cost, not both. The cost may include key establishment, signature verification, parsing, proxy work, and communication.
Since the indicator is symmetric, antiparallel directed interactions can be summed when an undirected cut or labeling algorithm is used.

\begin{figure}[tbp]
  \centering
  \begin{tikzpicture}[>=Latex]
  \node[svc] (a) {$a$};
  \node[svc,below=10mm of a] (b) {$b$};
  \node[svc,below=10mm of b] (c) {$c$};

  \node[svc,right=34mm of a] (d) {$d$};
  \node[svc,below=10mm of d] (e) {$e$};
  \node[svc,below=10mm of e] (f) {$f$};

  \node[domainbox,fit=(a)(b)(c),label=above:{Domain 1}] (box1) {};
  \node[domainbox,fit=(d)(e)(f),label=above:{Domain 2}] (box2) {};

  \draw[intedge] (a) -- (b);
  \draw[intedge] (b) -- (c);
  \draw[boundedge] (c) -- (d);
  \draw[boundedge] (b) -- (e);
  \draw[intedge] (e) -- (f);

  \node[
    anchor=west,
    draw,
    rounded corners,
    fill=white,
    inner sep=3pt,
    font=\scriptsize,
    align=left,
  ] (leg) at ($(box1.south)!0.5!(box2.south)+(0,-8mm)$) {%
    \tikz{\draw[intedge] (0,0) -- (7mm,0);} \; interaction edge ($E$)\\
    \tikz{\draw[boundedge] (0,0) -- (7mm,0);} \; boundary edge ($D(u)\ne D(v)$)%
  };
\end{tikzpicture}
  \caption{\label{fig:running_example_main}
  The two red service-call arrows cross the fixed domains, and their $r_{uv}\ell_{uv}c_{\mathrm{pqc}}(u,v)$ contributions sum to $\mathrm{Lat}(D)=220$ in the unit-normalized toy calibration. Blue arrows remain within a domain.
  Supplementary Sec.~\ref{SIsec:worked_example} gives the derivation-tree comparison at this fixed cut.}
\end{figure}
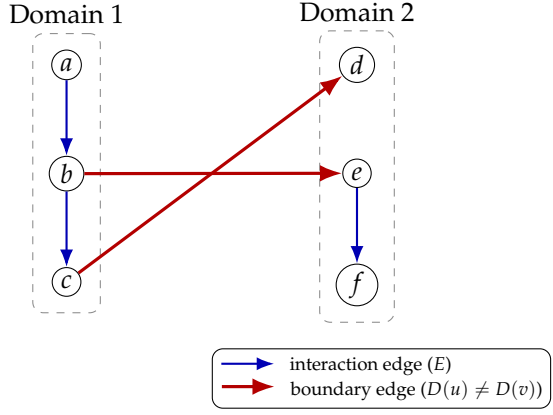
At this fixed cut, raising the fanout bound from $\Delta=2$ to $3$ permits direct root issuance and reduces $\mathrm{BR}_{\mathrm{node}}$ from $0.06$ to $0.03$ without changing the boundary latency.

Stateful stacks can be retained without changing the optimizer by calibrating
\begin{equation}
c_{\mathrm{eff}}(u,v)
=
c^{\mathrm{warm}}_{uv}
+
P^{\mathrm{cold}}_{uv}c^{\mathrm{cold}}_{uv},
\end{equation}
where $c^{\mathrm{warm}}_{uv}$ is the reused-session cost, $c^{\mathrm{cold}}_{uv}$ is the additional setup penalty, and $P^{\mathrm{cold}}_{uv}$ is the probability of incurring that penalty under the workload and reuse policy.
Request-class mixtures and critical-path weights similarly alter edge coefficients. Supplementary Sec.~\ref{SIsup:deployment} gives these extensions.

Our named TLS measurements evaluate three cryptographic profiles on five controlled network paths.
Under mutual authentication, 20 requests per connection, and a $0.05$ full-handshake probability, the hybrid X25519+ML-KEM-768/ML-DSA-65 profile spans \PQCAOneCostLow--\PQCAOneCostHigh\,ms per crossing call from the unshaped local path to the constrained lossy path.
These values are architecture inputs rather than algorithm rankings: changing the implementation, platform, path, or reuse policy changes the edge weights.
Supplementary Table~\ref{SItab:pqc_named_profile} gives the profiles, path conditions, and measured costs.
\section{Optimization problem}\label{sec:optimization}
The design object is $(D,\{H_i\})\in\mathcal{F}$, and its Pareto objective is
\begin{equation}
\min_{(D,\{H_i\})\in\mathcal{F}}
\bigl(\mathrm{Lat}(D),\mathrm{BR}_{\mathrm{node}}(D,\{H_i\})\bigr).
\end{equation}
Boundary placement changes root exposure and boundary latency. The derivation structures determine delegated-principal propagation inside each domain.
Fixing $D$ separates the derivation-tree best responses. A latency-first optimization can miss the optimum because the assignment changes $\Phi_i(V_i)$ and can change feasibility through $(\Delta,h)$.
Supplementary Sec.~\ref{SIsec:worked_example} gives a three-service counterexample.

For architecture decisions we use a latency budget $B$:
\begin{equation}
\begin{aligned}
\min_{(D,\{H_i\})\in\mathcal{F}}\quad&
\mathrm{BR}_{\mathrm{node}}(D,\{H_i\})\\
\text{s.t.}\quad&\mathrm{Lat}(D)\le B.
\end{aligned}
\label{eq:budgetopt}
\end{equation}
Varying $B$ traces the containment--latency frontier without mixing units.
Additional linear constraints, such as a boundary-byte cap, can be imposed in the same form.

A raw-unit linear scalarization is
\begin{equation}
\min_{(D,\{H_i\})\in\mathcal{F}}
\mathrm{BR}_{\mathrm{node}}(D,\{H_i\})
+\lambda\,\mathrm{Lat}(D),\qquad \lambda\ge0.
\label{eq:mainopt}
\end{equation}
Here $\lambda$ converts boundary latency to impact units. When latency is measured in ms/request, its units are impact per (ms/request).
To sample tradeoffs independently of reporting units, candidate generation uses
\begin{equation}
\alpha\frac{\mathrm{Lat}(D)}{L_0}
+(1-\alpha)\frac{\mathrm{BR}_{\mathrm{node}}(D,\{H_i\})}{R_0},
\qquad \alpha\in[0,1],
\label{eq:normalized_scalarization}
\end{equation}
where $L_0=\sum_{(u,v)\in E}r_{uv}\ell_{uv}c_{\mathrm{pqc}}(u,v)$ is the cost if every modeled interaction crossed a boundary and $R_0=\sum_{v\in V}w(v)$ is total modeled impact.
If either scale is zero, its identically zero term is omitted.
For positive $L_0,R_0$ and $\alpha<1$, Eq.~(\ref{eq:normalized_scalarization}) is equivalent to Eq.~(\ref{eq:mainopt}) with $\lambda=\alpha R_0/((1-\alpha)L_0)$.
Thus changing milliseconds to microseconds or multiplying all service weights by a common constant leaves candidate ordering unchanged.

Linear scalarization generates supported candidates.
For each budget $B$, the implementation refines retained scalarized and traffic-structured candidates using moves that preserve $\mathrm{Lat}(D)\le B$. It then selects the feasible candidate with minimum blast radius.
This $\varepsilon$-constraint refinement addresses Eq.~(\ref{eq:budgetopt}) and can recover unsupported points. It remains heuristic in the general regimes.

\section{Shared-issuer validation and rescoring}\label{sec:explicit_prop_setting}
Shared-issuer acceptance adds deployment-specific propagation to the core design problem.
The domain assignment and derivation trees remain the architecture variables, while the observed issuer-minting and verifier-acceptance relations $(M,J)$ determine whether a candidate's blast radius can be evaluated by the core score.
When these relations create additional entry points, final candidate selection must use explicit issuer reachability.

Let $I$ be issuers with compromise probabilities $p(a)$, $E_{\mathrm{auth}}\subseteq E$ protected calls, $M\subseteq V\times I$ caller--issuer minting, and $J\subseteq I\times V$ verifier acceptance.
Issuer events in $I$ are distinct from service and domain-root events. A physical principal occupying several roles is represented by one compromise event, with its reachable targets unioned before scoring.
A redesign remains authenticatable only if
\[
\forall (u,v)\in E_{\mathrm{auth}},\quad
\exists a\in I:\ (u,a)\in M\ \land\ (a,v)\in J.
\]
Compromise of issuer $a$ reaches
\begin{equation}
R_{H,J}(a)=
\bigcup_{(a,v)\in J}
\{u\in V_{D(v)}:v\leadsto u\text{ in }H_{D(v)}\}.
\label{eq:explicit_reach}
\end{equation}
Equation~(\ref{eq:explicit_reach}) conservatively treats successful impersonation at an accepted target $v$ as access to every downstream credential capability represented below $v$ in $H_{D(v)}$.
Because verifier-acceptance relations are fixed deployment inputs rather than optimization variables, feasible redesigns preserve required/allowed acceptance pairs and enforce intended issuer isolation through (C2).

For issuer $a$ and domain $i$, $T_i(a)=\{v\in V_i:(a,v)\in J\}$ is the set of accepted targets.
Define
\[
W_{H_i}(S)=\sum_{u:\exists s\in S,\ s\leadsto u\text{ in }H_i}w(u).
\]
Issuer $a$ contributes $p(a)W_{H_i}(T_i(a))$ in domain $i$.
The full first-order score is therefore
\begin{equation}
\begin{aligned}
&\mathrm{BR}_{\mathrm{explicit}}(D,\{H_i\},J)\\
&\quad=\mathrm{BR}_{\mathrm{node}}(D,\{H_i\})\\
&\qquad+\sum_{a\in I}p(a)\sum_{i=1}^{k}W_{H_i}(T_i(a)).
\end{aligned}
\label{eq:explicit_br}
\end{equation}
An additive alternative assigns issuer risk separately to each accepted target and contributes
$p(a)\sum_{v\in T_i(a)}W_{H_i}(\{v\})$.
We call the result the additive issuer score. It can count the same descendant through several accepted targets.
Such overlap is absent when the targets form an antichain, meaning that no accepted target is an ancestor of another.
A star refers to the derivation arborescence $H_i$, where $\rho_i$ is the parent of every service. The interaction graph $G$ remains arbitrary.

\begin{proposition}[Additive issuer-score upper bound]\label{prop:explicit_compiled_upper_bound}
For fixed $(D,\{H_i\})$ and nonnegative $p,w$, the additive issuer score upper-bounds the first-order explicit issuer-reachability score.
Equality holds if, for every issuer and domain, the accepted targets form an antichain in $H_i$.
\end{proposition}
The bound follows because weighted union size is at most the sum of weighted set sizes. Antichain targets in an arborescence have disjoint descendant sets.
A full proof appears in Supplementary Sec.~\ref{SIapp:proof_explicit_compiled_upper_bound}.

\begin{corollary}[Star-overlay exactness]\label{cor:explicit_star_exact}
For star $H_i$, the additive issuer score is exact for any number of accepted targets.
\end{corollary}
\begin{corollary}[Chain-overlay exactness]\label{prop:restricted_chain_overlay_exact}
For chain $H_i$, the additive issuer score is exact when each issuer has at most one accepted target per domain, using
\[
p_{\mathrm{eff}}(v)=p(v)+\sum_{a:(a,v)\in J}p(a).
\]
The objectives therefore agree on any candidate family satisfying this condition.
\end{corollary}
Both specializations follow in the proof of Proposition~\ref{prop:explicit_compiled_upper_bound} in Supplementary Sec.~\ref{SIapp:proof_explicit_compiled_upper_bound}.
For a fixed deeper tree, removing any accepted target reachable from another accepted target leaves explicit reach unchanged.
This deduplication depends on $H_i$ and cannot be represented by a candidate-independent transformation to $p_{\mathrm{eff}}$.
If accepted targets overlap by ancestry, the additive issuer score is only a screening upper bound and final selection requires explicit scoring.

\section{Complexity and structural regimes}\label{sec:complexity}
The joint problem ranges from standard cut primitives to coupled NP-hard cases. Write $n=|V|$.
The scalarized direct-issuance $k=2$ slice is a minimum cut. For $k\ge3$, multiway-cut variants are NP-hard.\cite{Dahlhaus1992MultiwayCuts}
Derivation constraints can preserve hardness even when boundary latency vanishes.

\begin{theorem}[Coupled NP-hardness (chains)]\label{thm:coupled_hard_chain}
Let $E=\emptyset$, impose $(\Delta,h)=(1,|V|)$ with complete derivation eligibility, and set $p(\rho_i)=0$ for every domain.
Minimizing $\mathrm{BR}_{\mathrm{node}}$ over feasible designs is NP-hard for every fixed $k\ge3$, and strongly NP-hard when $k$ is part of the input.
\end{theorem}
The reduction is from $P_k\mathbin{\|}\sum_jw_jC_j$, where $C_j$ is job completion time: domains are machines, chain order is job order, and $p(v_j)$ is a scaled processing time.
The full proof appears in Supplementary Sec.~\ref{SIapp:proof_coupled_hard_chain}.

Under direct issuance, each root issues credentials directly to its domain's services, so each $H_i$ is a star. The assignment-dependent risk is then the root-risk contribution, which gives the unary costs in cut and labeling formulations.
An anchor is a service fixed to a specified domain label.
\begin{proposition}[Two-domain min-cut regime]\label{prop:coupled_min_cut_star}
Consider Eq.~(\ref{eq:mainopt}) with $k=2$ and label-specific root probabilities independent of $D$. Assume every root-to-service arc is eligible, $\Delta\ge |V|-1$, and $h\ge1$. No additional domain-policy constraints apply.
Exact optimization under (C1) is the minimum over ordered anchor pairs of $s$--$t$ cuts with service-label costs $w(v)p(\rho_i)$ and edge-disagreement costs $\lambda r_{uv}\ell_{uv}c_{\mathrm{pqc}}(u,v)$, and remains polynomial time.
\end{proposition}
Supplementary Sec.~\ref{SIapp:proof_coupled_min_cut_star} gives the construction.
The proposition maps the architecture variables to cut costs. Polynomial-time optimization then follows from the standard minimum-cut algorithm.
The two-domain supported-point procedure varies the scalarization parameter and collects the resulting cut solutions. Supplementary Corollary~\ref{SIcor:parametric_cut_sweep} states the corresponding parametric-flow result.
A hub is a selected depth-one node in a depth-two derivation tree.
The solver routes require both semantic and topological eligibility.
Direct root-to-workload issuance gives a star $H_i$. A root--intermediate--workload hierarchy gives depth two, and sequential delegation gives a chain. Shared verifier acceptance requires the explicit-propagation model rather than an arborescence shortcut.
Table~\ref{tab:regime_map} summarizes the applicable solver choices.
The chain route uses exact ratio ordering, while the multiway-cut, metric-labeling, and parametric rows inherit guarantees from established algorithms.\cite{Dahlhaus1992MultiwayCuts,KleinbergTardos2002,GalloGrigoriadisTarjan1989}
The corresponding reductions and bounded-treewidth statement appear in Supplementary Sec.~\ref{SIsup:proofs}.

\begin{table*}[tbp]
\caption{\label{tab:regime_map}
Solver guarantees by structural regime.
Coupled chain and multiway cases are hard, while the listed two-domain, bounded-width, and fixed-partition cases admit exact methods under their stated assumptions.
  Under the antichain condition, additive issuer scoring equals explicit issuer-reachability scoring. Supporting statements appear in Supplementary Sec.~\ref{SIsup:proofs}.}
\centering
\small
\setlength{\tabcolsep}{3.5pt}
\renewcommand{\arraystretch}{1.16}
\begin{tabular}{>{\raggedright\arraybackslash}p{0.24\textwidth} >{\raggedright\arraybackslash}p{0.20\textwidth} >{\raggedright\arraybackslash}p{0.25\textwidth} >{\raggedright\arraybackslash}p{0.23\textwidth}}
\rowcolor{black!15}
\textbf{Applicable structure} & \textbf{Method} & \textbf{Guarantee} & \textbf{What it provides}\\
\hline\hline
\rowcolor{black!10}
\multicolumn{2}{l}{\textbf{Coupled partition and derivation ($D,H$)}} & & \\[-0.2ex]
Chain, $k\ge3$ & Parallel-machine scheduling reduction & \textbf{Hard.} NP-hard for fixed $k\ge3$ and strongly NP-hard when $k$ varies. & Coupled hardness boundary.\\[1pt]
\rowcolor{black!5}
Direct issuance, $k=2$ & Anchored min-cut and parametric flow & \textbf{Exact.} $O(n^2)$ anchor pairs. & Boundary and supported Pareto-point optimization.\\[1pt]
Direct issuance with equal root risk and $k$ anchors & Multiway cut & \textbf{Hard.} NP-hard with a $(2-2/k)$-approximation. & Anchored boundary optimization.\\[1pt]
\rowcolor{black!5}
Direct issuance with label-specific root risk & Graph labeling (Potts/metric labeling) & \textbf{Conditional.} Exact for fixed $k$ on bounded-treewidth interaction graphs. Specified anchors enforce nonempty labels. A $2$-approximation applies only when labels may be unused. & Label-aware boundary optimization.\\
\hline
\rowcolor{black!10}
\multicolumn{2}{l}{\textbf{Fixed partition (optimize $H_i$)}} & & \\[-0.2ex]
Chain with unrestricted ordering & Ratio ordering & \textbf{Exact.} $O(n\log n)$. & Optimal derivation order.\\[1pt]
\rowcolor{black!5}
Direct issuance & Root-to-service arcs & \textbf{Exact.} No inner search. & Derivation star.\\[1pt]
Depth two with fixed hubs & Capacitated assignment & \textbf{Exact} after fixed-$\Delta$ hub enumeration. & Small-fanout assignment.\\
\hline
\rowcolor{black!10}
\multicolumn{2}{l}{\textbf{Issuer overlay validation}} & & \\[-0.2ex]
Shared-issuer acceptance & Explicit reachability and antichain test & \textbf{Conditional.} Equal scores under the antichain condition. & Determines when explicit rescoring is required.\\
\end{tabular}
\end{table*}

SPIRE documents single-authority, nested-intermediate, and federated trust-domain deployments, while OAuth 2.0 Token Exchange represents impersonation and delegation chains.\cite{SPIREScaling,SPIFFEFederation,RFC8693}
These specifications establish implementability. Deployment prevalence remains unknown.
Service-call traces contain interaction topology and omit credential semantics. The replay therefore does not estimate regime prevalence.

\section{Operational solver}\label{sec:algorithms}
The implementation follows four steps.

\PaperRunIn{1) Build the instance}
Construct $G$ from policy-gated call telemetry, calibrate crossing costs on the deployed cryptographic stack, specify $w$ and $p$ for a planning horizon, and encode domain policy, derivation eligibility, fanout, and depth constraints.
Retain $(I,M,J)$ when shared-issuer acceptance creates entry points not represented by the domain roots.

\PaperRunIn{2) Generate candidates}
Use exact structural solvers when their assumptions hold. The joint routes are the two-domain cut for direct issuance and bounded-treewidth dynamic programming for eligible low-width graphs. After fixing the domain assignment, use an exact chain, direct-issuance, or fixed-hub depth-two routine when applicable.
Supplementary Tables~\ref{SItab:regime_applicability} and~\ref{SItab:formal_refinement_stress} report the tested low-width dispatch and the treewidth increase caused by scoped refinement.
Otherwise initialize $D$ from traffic structure using spectral or multilevel partitioning,\cite{KarypisKumar1998Metis} construct a feasible derivation tree in each domain, and alternate risk-aware boundary moves with derivation updates.\cite{KernighanLin1970,FiducciaMattheyses1982}
The default general-tree construction sorts by $p(v)/w(v)$ and fills a $\Delta$-ary tree breadth first. It is exact for chains and heuristic otherwise.

\PaperRunIn{3) Refine under the budget}
Weighted-sum solutions and deterministic traffic partitions form the initial candidate set.
For each latency budget $B$, apply only moves and swaps that preserve feasibility and satisfy $\mathrm{Lat}(D)\le B$, then retain the candidate with the lowest blast-radius score.
This hard-budget stage can recover designs omitted by linear scalarization. Increasing the density of the $\alpha$ grid alone cannot provide that coverage.
Supplementary Table~\ref{SItab:solver_budget_path_ablation} isolates the contribution of this stage.

\PaperRunIn{4) Validate semantics and cost}
Use the additive issuer score directly when the antichain condition holds. Otherwise, trace issuer reachability for each retained candidate before selection.
Exact enumeration provides reference solutions on small instances. Larger instances use additive or proxy scores to build a shortlist, then evaluate that shortlist using explicit reachability.
For chain derivations, the proxy orders services using $p_{\mathrm{eff}}$ and evaluates the resulting order under both derivation and acceptance paths.
After choosing $B$ from measured latency constraints, rank designs that satisfy it by blast radius and inspect the edge and compromise-point contributions.
Then remeasure the selected crossing edges in separate validation blocks. Reject a design whose observed boundary latency exceeds $B$, update the calibration or reserve, and rerun the budgeted selection.
The reserve is a planning control rather than a statistical guarantee.
Control-plane issuer inventories, trust stores, and delegation metadata determine plausible $H_i$ and $J$. Where these are uncertain, solve several defensible scenarios rather than treating one inferred graph as ground truth.
Supplementary Sec.~\ref{SIsup:deployment} gives the extraction and refinement details.

\PaperRunIn{Guarantees and limits}
For a fixed domain assignment, the implementation contains exact chain, direct-issuance, and fixed-hub depth-two derivation routines. It also contains exact two-domain direct-issuance/min-cut and anchored low-treewidth assignment routines.
The general candidate search assumes complete within-domain derivation eligibility. Restricted relations can be represented and checked, while optimizing over them requires a separate eligible-tree routine.
The implementation uses anchored $\alpha$-expansion as a graph-labeling heuristic and generates supported points through repeated exact two-domain solves. The $2$-approximation and parametric-flow guarantees in Table~\ref{tab:regime_map} refer respectively to the algorithms of Kleinberg--Tardos and Gallo--Grigoriadis--Tarjan.\cite{KleinbergTardos2002,GalloGrigoriadisTarjan1989}
The implementation applies one propagation model to the whole instance and does not yet combine different models or solvers across local graph regions.
General guarantees for explicit propagation and such locally mixed strategies remain open.
\section{Evaluation}\label{sec:evaluation}
We ask three questions: whether the coupled objective changes the design relative to a staged baseline, whether designs selected using measured crossing costs retain their benefit and latency compliance on held-out measurements, and when shared-issuer relationships require explicit rescoring.

\PaperRunIn{Setup}
The main trace replay applies the optimization to a Train-Ticket service-interaction graph with $\TrainTicketServices$ services and $\TrainTicketEdges$ directed edges, extracted from $\TrainTicketTraces$ public Jaeger traces.\cite{Steidl2022TrainTicketAnomalies}
Edge multiplicities determine normalized calls/request. The baseline uses $\ell_{uv}=1$. Supplementary Sec.~\ref{SIsup:train_ticket_instantiation} defines the scenario inputs $w(v)$ and $p(x)$.
The heterogeneous-baseline scenario retains these priors. The cluster-skew scenario triples $p(v)$ in the largest supplied service cluster, capped at $0.25$.
The formal problem treats the number of domains $k$ as an input. Replay searches $k\in\{1,\ldots,6\}$ under a computational cap. A value of $k=6$ means the largest tested count, not an optimum over larger $k$.
Replay root priors are zero and no fixed per-domain management charge is applied, so these experiments isolate delegated-principal risk and crossing cost. An equal nonzero root prior would add the same partition-invariant term to every candidate.

Unless stated otherwise, replay assumes complete within-domain derivation eligibility, and candidate partitions satisfy the domain-size condition induced by $(\Delta,h)=(3,3)$.
Each partition is scored with the deterministic breadth-first construction used for candidate generation.
Within each domain, services are ordered by $p(v)/w(v)$ and attached in that order to the earliest parent with remaining fanout, level by level to depth $h$.
Thus, these replay results optimize one feasible family under complete eligibility. Global optimization over all admissible arborescences and replay under deployment-specific eligibility remain outside the evaluation.
The controlled solver-transfer and baseline replay experiments retain $c_{\mathrm{pqc}}=0.03$\,ms per crossing call and 20 calls per top-level request so that only the tested algorithmic factor changes.
The deployment study instead uses \PQCMeasurementObservations{} TLS observations from \PQCProfileCount{} named cryptographic profiles and \PQCNetworkCount{} network paths, then optimizes with those measured profiles and validates selected crossings on held-out blocks.
We denote classical X25519/ECDSA-P256 by C0, hybrid X25519+ML-KEM-768/ECDSA-P256 by K1, and hybrid X25519+ML-KEM-768/ML-DSA-65 by A1. Paths N0--N4 denote the unshaped local, datacenter, regional, edge or mobile, and constrained lossy profiles.
In the replay text and tables, $\mathrm{BR}$ abbreviates the conservative $\mathrm{BR}_{\mathrm{node}}$ score, and $\mathrm{BR}_0$ is its $k=1$ value under the same workload, risk scenario, derivation method, and feasibility filter.
Candidate selections use fixed seeds and pinned numerical dependencies. Reported wall-clock measurements remain machine dependent. Supplementary Sec.~\ref{SIsup:eval_protocol} gives parameters and source records.

\PaperRunIn{Q1: Does joint optimization change the design?}
Table~\ref{tab:joint_vs_staged} compares joint optimization with the staged baseline on controlled $n=9$, $k=3$ instances. For each instance and budget, the exact reference evaluates every feasible assignment of the nine services to three nonempty labeled domains.
The staged baseline minimizes the direct-issuance blast-radius score under the latency budget. Each seed assigns the three domain labels deterministic heterogeneous root priors drawn uniformly from $[0.01,0.07)$. These priors remain fixed across partitions, budgets, and risk settings. If several assignments tie, the baseline receives the one with the lowest chain or depth-two score. This tie rule favors the staged baseline by giving it information that a practical sequential procedure would not have.
Joint optimization minimizes the chain or depth-two score directly under the same budget.

\begin{table*}[tbp]
  \caption{\label{tab:joint_vs_staged}
  Joint optimization lowers the chain or depth-two blast-radius score relative to the staged baseline in \JointStagedWorseCases{} of \JointStagedWorseTotal{} feasible comparisons.
  Each comparison enumerates every feasible nonempty domain assignment. Regret is the staged score minus the exact joint optimum, divided by that optimum.
  H is the heterogeneous-baseline risk setting and S triples service risk in the largest supplied cluster with clipping at $0.25$.
  Each setting contains 20 seeds and three budgets. Five depth-two cases per setting have no feasible $k=3$ design at the tightest budget and are omitted.}
  \centering
  \footnotesize
  \PaperTableRows
  \begin{tabular}{llrrrr}
Derivation & Risk & Cases & Staged worse (\%) & Mean regret (\%) & Max regret (\%) \\
\hline\hline
Chain & H & 60 & 98 & 29.38 & 66.41 \\
Chain & S & 60 & 97 & 32.73 & 64.67 \\
Depth two & H & 55 & 64 & 2.47 & 10.75 \\
Depth two & S & 55 & 78 & 2.67 & 14.72 \\
\end{tabular}

\end{table*}

The staged score is higher in nearly every chain case: \JointChainStagedWorseCases{} of \JointChainStagedWorseTotal{}.
It is also higher in most depth-two cases: \JointDepthStagedWorseCases{} of \JointDepthStagedWorseTotal{}.
The larger chain regret shows that a partition chosen under direct issuance can be poorly matched to the delegated compromise paths introduced later.
The depth-two effect is smaller and remains present in most feasible cases.
Supplementary Sec.~\ref{SIsup:eval_protocol} reports the adaptive supported-point solver-call check, hard-budget refinement, low-width dispatch, and heuristic comparisons with exact small-instance references.

\PaperRunIn{Q2: Workload, risk, and crossing cost}
A reference replay first tests sensitivity to the service-risk scenario. At $B=0.10$\,ms, both scenarios select $k=\TrainDomains$.
The resulting $\mathrm{BR}/\mathrm{BR}_0$ is \TrainHRatio{} under the heterogeneous baseline and \TrainSRatio{} under cluster skew.
Under each scenario, the corresponding domain count recurs in at least $\TrainCountMatch\%$ of edge-count resamples and $\TrainSolverMatch\%$ of solver seeds (Supplementary Sec.~\ref{SIsup:recommendation_stability}).

The operational comparison keeps the bounded breadth-first derivation family fixed and changes only how the domain assignment is selected.
Risk-aware search optimizes the conservative score under the measured latency budget.
For each $k\in\{1,\ldots,6\}$, latency-first search applies the restart-limited latency-only search and retains its lowest-latency partition. It then scores these partitions with the fixed breadth-first derivation family and selects the lowest-score candidate satisfying the budget. Traffic clustering partitions the interaction graph without risk input, and the single-domain design provides the baseline.
All methods use the same edge-specific calibration-block maxima with a $5\%$ budget reserve. Their selected crossings are then measured in held-out blocks.

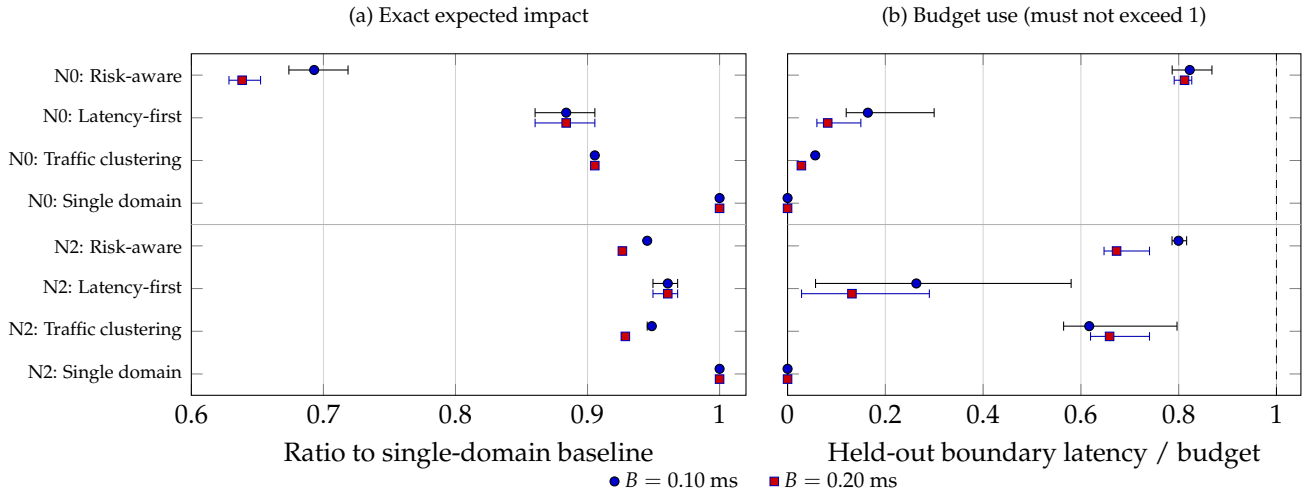
\begin{figure*}[tbp]
  \centering
  \begin{tikzpicture}
\begin{axis}[
  name=pqcmethodimpact,
  width=0.49\textwidth,
  height=6.1cm,
  xmin=0.60,
  xmax=1.02,
  ymin=-0.5,
  ymax=7.5,
  ytick={0,1,2,3,4,5,6,7},
  yticklabels={N2: Single domain,N2: Traffic clustering,N2: Latency-first,N2: Risk-aware,N0: Single domain,N0: Traffic clustering,N0: Latency-first,N0: Risk-aware},
  yticklabel style={font=\scriptsize,align=right},
  xlabel={Ratio to single-domain baseline},
  title={(a) Exact expected impact},
  title style={font=\footnotesize},
  xmajorgrids=true,
  grid style={line width=.1pt,draw=gray!35},
]
\addplot+[only marks,mark=*,mark size=1.7pt,color=black,
  restrict expr to domain={\thisrow{budget_ms}}{0.099:0.101},
  error bars/.cd,x dir=both,x explicit]
  table[x=br_exact_mean_ratio,y=plot_y,x error minus=br_exact_error_minus,x error plus=br_exact_error_plus,col sep=comma]
  {data/fig_inputs/pqc_method_comparison.csv};
\addplot+[only marks,mark=square*,mark size=1.6pt,color=blue!70!black,
  restrict expr to domain={\thisrow{budget_ms}}{0.199:0.201},
  error bars/.cd,x dir=both,x explicit]
  table[x=br_exact_mean_ratio,y=plot_y,x error minus=br_exact_error_minus,x error plus=br_exact_error_plus,col sep=comma]
  {data/fig_inputs/pqc_method_comparison.csv};
\addplot[gray!60,thin,forget plot] coordinates {(0.60,3.5) (1.02,3.5)};
\end{axis}

\begin{axis}[
  name=pqcmethodlatency,
  at={(pqcmethodimpact.east)},
  anchor=west,
  xshift=0.55cm,
  width=0.46\textwidth,
  height=6.1cm,
  xmin=0,
  xmax=1.05,
  ymin=-0.5,
  ymax=7.5,
  ytick={0,1,2,3,4,5,6,7},
  yticklabels={,,,,,,,,},
  xlabel={Held-out boundary latency / budget},
  title={(b) Budget use (must not exceed 1)},
  title style={font=\footnotesize},
  xmajorgrids=true,
  grid style={line width=.1pt,draw=gray!35},
  legend to name=pqcmethodlegend,
  legend columns=2,
  legend style={draw=none,fill=none,font=\footnotesize,/tikz/every even column/.append style={column sep=1em}},
]
\addplot+[only marks,mark=*,mark size=1.7pt,color=black,
  restrict expr to domain={\thisrow{budget_ms}}{0.099:0.101},
  error bars/.cd,x dir=both,x explicit]
  table[x=latency_budget_mean_ratio,y=plot_y,x error minus=latency_budget_error_minus,x error plus=latency_budget_error_plus,col sep=comma]
  {data/fig_inputs/pqc_method_comparison.csv};
\addlegendentry{$B=0.10$ ms}
\addplot+[only marks,mark=square*,mark size=1.6pt,color=blue!70!black,
  restrict expr to domain={\thisrow{budget_ms}}{0.199:0.201},
  error bars/.cd,x dir=both,x explicit]
  table[x=latency_budget_mean_ratio,y=plot_y,x error minus=latency_budget_error_minus,x error plus=latency_budget_error_plus,col sep=comma]
  {data/fig_inputs/pqc_method_comparison.csv};
\addlegendentry{$B=0.20$ ms}
\addplot[gray!60,thin,forget plot] coordinates {(0,3.5) (1.05,3.5)};
\addplot[black,dashed,thin,forget plot] coordinates {(1,-0.5) (1,7.5)};
\end{axis}
\node[anchor=north] at ($(pqcmethodimpact.south)!0.5!(pqcmethodlatency.south)+(0,-0.75cm)$)
  {\pgfplotslegendfromname{pqcmethodlegend}};
\end{tikzpicture}
  \caption{\label{fig:pqc_method_comparison}
  Risk-aware search produces the largest reduction in exact expected impacted weight on N0. On N2, its observed range overlaps that of traffic clustering.
  N0 is the unshaped local path. N2 has 35-ms round-trip time, 100-Mbit/s rate, and 0.1\% packet loss.
  Panel (a) reports exact expected impacted weight under independent root and service compromise events. Candidate selection uses the conservative linear score.
  Panel (b) divides held-out boundary latency by the budget $B$ in ms/request. The dashed line marks the budget limit.
  Points are means over five independently calibrated runs and bars span the observed range.
  The comparison tests domain selection within one fixed derivation family. It does not establish global optimization over all derivation arborescences.}
\end{figure*}

On N0, latency-first search has a mean exact expected-impact ratio of \PQCNZeroStagedExact{} at both budgets.
Risk-aware search lowers these ratios to \PQCNZeroRiskExactTight{} and \PQCNZeroRiskExactLoose{}, respectively.
Every risk-aware design satisfies its held-out budget in these runs.
The broader five-path sensitivity, including absolute C0/K1/A1 costs and incremental A1-minus-C0 migration costs, appears in Supplementary Fig.~\ref{SIfig:pqc_path_architecture} and Table~\ref{SItab:pqc_profile_architecture}.

Table~\ref{tab:pqc_budget_validation} reports five independent graph-replay runs on each of N0 and N2 that separated calibration blocks from held-out validation blocks and checked two budgets per run.
Reserving $5\%$ of the budget under either block-maximum rule meets all \PQCHeadroomChecksPerRule{} path-budget checks per rule.
Finite held-out success does not establish a path-independent guarantee, so remeasurement and rejection remain necessary.

\begin{table}[tbp]
  \caption{\label{tab:pqc_budget_validation}
  Held-out budget compliance across five independent runs and two budgets per path.
  Both block-maximum rules meet every check with a $5\%$ reserve. Point estimates and unmodified maxima fail on some paths.
  N0 is unshaped. N2 has 35-ms round-trip time, 100-Mbit/s rate, and 0.1\% packet loss. $B$ is in ms/request.
  Standalone uses a path-level profile, graph-weighted uses one call-weighted coefficient, and edge-specific retains per-edge coefficients.
  Block maxima are taken across calibration blocks. Reserve rows inflate them by $1/(1-0.05)$.}
  \centering
  \footnotesize
  \PaperTableRows
  \begin{tabular}{lcc}
    \toprule
    Calibration & N0 & N2 \\
    \midrule
    Standalone profile & 2/10 & 10/10 \\
Graph-weighted mean & 1/10 & 10/10 \\
Edge-specific means & 1/10 & 9/10 \\
Graph block maximum & 7/10 & 10/10 \\
Edge-specific block maxima & 10/10 & 8/10 \\
Graph maximum + 5\% reserve & 10/10 & 10/10 \\
Edge maxima + 5\% reserve & 10/10 & 10/10 \\
\hiderowcolors
\bottomrule

  \end{tabular}
\end{table}

Additional solver-scaling, protocol, and calibration diagnostics appear in Supplementary Sec.~\ref{SIsup:eval_tables}.

\PaperRunIn{Q3: Explicit-rescoring conditions}
The hand-constructed shared-issuer reachability example in Fig.~\ref{fig:mini_jwt_explicit} uses unit-normalized interaction coefficients. Its edge labels are illustrative toy costs rather than measured latency.
Keeping that interaction cut fixed while adding acceptance from the isolated issuer \texttt{auth} to \texttt{payments} and \texttt{orders} leaves boundary latency unchanged and adds explicit propagation paths.
The discrepancy comes from semantics: arborescence-only scoring omits issuer entry paths that the deployment accepts.

On the larger candidate sets derived from traces, ranking by the additive issuer score alone identifies the best explicitly evaluated candidate in \OverlayCompiledRecoveries{} of \OverlayComparisons{} comparisons and misses by at most \OverlayMaxGap\%. Explicitly evaluating the three candidates ranked highest by the reachability-aware proxy recovers the best generated candidate in all \OverlayComparisons{} comparisons.
When the antichain condition holds, the additive issuer score can select the final design. In all other cases, final selection uses explicit issuer reachability over the shortlist.

\begin{figure}[tbp]
  \centering
  \includegraphics[width=\columnwidth]{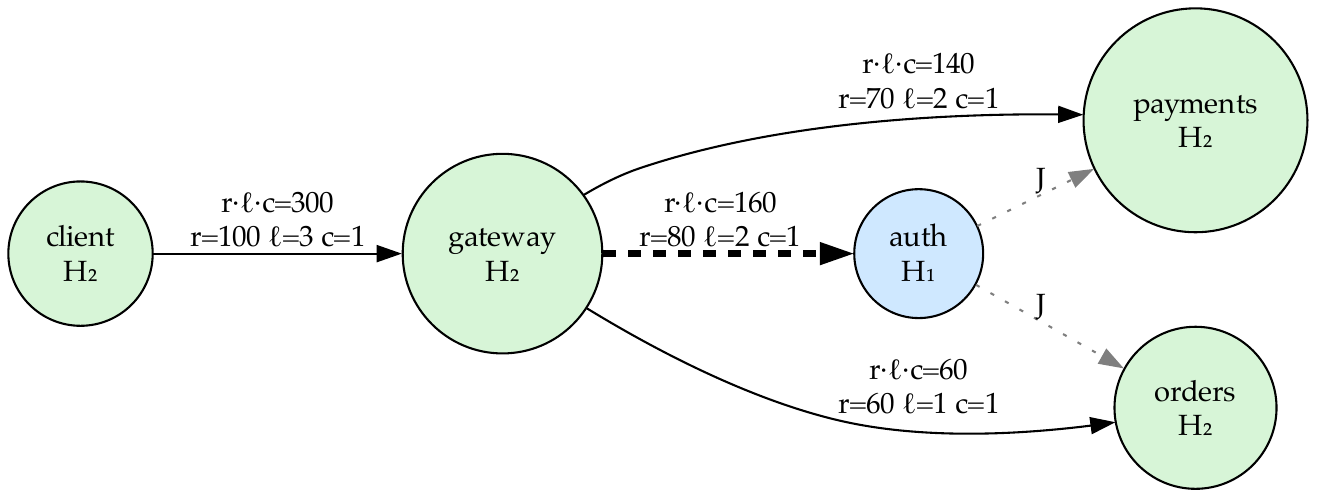}
  \caption{\label{fig:mini_jwt_explicit}
  Shared-issuer reachability example showing why boundary placement alone does not capture issuer propagation.
  The $H_1/H_2$ labels identify the credential-derivation trees containing each principal.
  Black arrows are service interactions. Solid arrows stay within $H_2$, whereas the thick dashed \texttt{gateway}$\to$\texttt{auth} arrow is the sole boundary crossing and contributes the full toy cut cost $r\ell c=160$.
  Gray dotted arrows are verifier-acceptance pairs $J$, not derivation arcs.
  Because \texttt{payments} and \texttt{orders} in $H_2$ accept credentials from \texttt{auth} in $H_1$, issuer compromise reaches both services without changing $\mathrm{Lat}(D)$.}
\end{figure}

\section{Limitations}\label{sec:limitations}
The model is a credential-architecture planning abstraction. It does not cover every enterprise compromise path.
Its conclusions are conditional on the supplied workload, criticality, compromise, and cost scenarios.
\begin{itemize}
\setlength{\itemsep}{0.35\baselineskip}
\item Propagation is limited to credential authority. Verifier compromise, software supply-chain compromise, and non-credential exploitation paths are not modeled.
Compromise probabilities are exogenous, although segmentation may change real exposure.
\item The core assumes single-parent derivation.
Multi-parent or threshold issuance requires a directed acyclic graph (DAG) model.
Explicit issuer acceptance is handled separately. General complexity guarantees and globally coordinated mixtures of core and explicit-propagation scoring remain open.
\item The mapping from trust boundaries to crossing cost assumes key-establishment domains in which intra-domain payload protection uses established symmetric keys or derived credentials.
Deployments that perform independent public-key operations within a domain must represent those operations as additional weighted events or refine the domain assignment.
\item Boundary latency is linearized from measured or estimated effective crossing costs.
Static edge coefficients can absorb path-specific mean effects. They do not directly model correlated packet loss, retries, shared congestion, queueing, or end-to-end tail latency.
The controlled TLS measurements expose path-specific retransmission and handshake effects only through effective coefficients. The graph replay estimates how several calibration and reserve rules transfer across N0 and N2. Finite held-out success cannot establish a path-independent guarantee.
Operational use therefore requires iterative calibration, optimization, and remeasurement of the selected crossing edges.
\item Trace results use one deterministic bounded-breadth-first derivation family and finite local-search budgets. The reported designs are best found within this search. Global optimality over all feasible arborescences remains unverified.
Supplementary Tables~\ref{SItab:solver_budget_path_ablation} and~\ref{SItab:joint_chain_transfer} quantify observed gaps against exhaustive references on small instances.
\item Two public trace-derived workloads, synthetic families, and controlled stack measurements do not establish production-wide generality.
Control-plane uncertainty should be represented by several plausible $H_i$ and $J$ scenarios.
\end{itemize}

\section{Related work}\label{sec:related}
The formulation combines a segmentation decision with credential-derivation design. Its compromise semantics differ from the two closest decision problems.
Authentication-graph partitioning and identity and access management (IAM) policy synthesis change access structure to reduce credential-connected components, unnecessary permissions, or compromise impact while limiting operational disruption.\cite{PopeTauritzKent2019,KazdagliTiwariKumar2022}
Their decision variables stop at access structure. The credential-authority arborescence among operational principals remains fixed.

Logical-key-hierarchy research chooses rooted auxiliary-key trees using update probabilities, communication cost, and network topology.\cite{WongGoudaLam1998KeyGraphs,WallnerHarderAgee1999RFC2627,ChanRajaramanSunZhu2009,SakaiYamamoto2005}
Those models use rekeying keys as internal vertices under a shared controller, and their objective is update or recovery communication.
Each domain has an independently compromised root. The hierarchy's internal vertices are operational principals, and steady-state service calls incur cost when they cross roots.

Hierarchical key assignment and delegation systems provide authority-structure context while optimizing different objectives.\cite{Atallah2009KeyHierarchies,Birgisson2014Macaroons,Andersen2019WAVE}
Zero-trust identity guidance defines relevant issuer relationships.\cite{NIST800207,WardBeyer2014BeyondCorp}
Risk-optimized and role-based microsegmentation methods synthesize access boundaries from policy or flow evidence while treating credential-derivation structure as fixed.\cite{NoelSwarupJohnsgard2023Microsegmentation,Mani2025ZTS}
The comparison turns on four elements used together in this formulation: an independent-root partition, an operational derivation forest, steady-state crossing cost, and a model-specific condition for exact additive issuer scoring.

\section{Conclusion}\label{sec:conclusion}
We formulate trust-boundary placement and credential delegation as one design problem: limit compromise reach without exceeding the latency budget for cross-boundary authentication. The model separates domain-authority compromise, delegated-principal compromise, and shared-issuer acceptance.
The formal results establish NP-hardness and identify polynomial-time cases under the stated direct-issuance assumptions: scalarized two-domain optimization and bounded-width interaction graphs with specified anchors or fixed $k$. The results also identify when additive issuer scoring equals explicit issuer-reachability scoring. Other regimes inherit guarantees from established scheduling, graph-labeling, and parametric-flow algorithms.

The evaluation shows where the formulation changes design decisions. Joint optimization lowers the chain or depth-two blast-radius score in \JointStagedWorseCases{} of \JointStagedWorseTotal{} exact comparisons, with the larger effect under chain delegation.
Measured path conditions change the domain count selected under the same latency budget. On N0, risk-aware search yields a large impact reduction. The five N2 runs leave the ordering of risk-aware search and traffic clustering unresolved. A $5\%$ reserve meets every held-out check for both tested block-maximum rules. The reserve remains a planning rule that requires deployment remeasurement.
The additive issuer score can misrank designs when accepted services have overlapping credential reach. In deployment, use the optimization to generate candidate architectures, trace shared-issuer reach where required, and remeasure selected crossings before accepting a design.

\section*{Data and Code Availability}
Code and processed data may be made available by the corresponding author upon reasonable request, subject to organizational approval.

\section*{Acknowledgment}
The authors thank K. Halunen for feedback on an earlier version of this manuscript. The views expressed in this paper are those of the authors and do not necessarily reflect the views and policies of their respective employers.

\PaperBalanceLastPage
\bibliographystyle{IEEEtran}
\bibliography{refs}

\begin{thebibliography}{10}
\providecommand{\url}[1]{#1}
\csname url@samestyle\endcsname
\providecommand{\newblock}{\relax}
\providecommand{\bibinfo}[2]{#2}
\providecommand{\BIBentrySTDinterwordspacing}{\spaceskip=0pt\relax}
\providecommand{\BIBentryALTinterwordstretchfactor}{4}
\providecommand{\BIBentryALTinterwordspacing}{\spaceskip=\fontdimen2\font plus
\BIBentryALTinterwordstretchfactor\fontdimen3\font minus
  \fontdimen4\font\relax}
\providecommand{\BIBforeignlanguage}[2]{{%
\expandafter\ifx\csname l@#1\endcsname\relax
\typeout{** WARNING: IEEEtran.bst: No hyphenation pattern has been}%
\typeout{** loaded for the language `#1'. Using the pattern for}%
\typeout{** the default language instead.}%
\else
\language=\csname l@#1\endcsname
\fi
#2}}
\providecommand{\BIBdecl}{\relax}
\BIBdecl

\bibitem{NIST800207}
\BIBentryALTinterwordspacing
S.~Rose, O.~Borchert, S.~Mitchell, and S.~Connelly, ``Zero trust
  architecture,'' National Institute of Standards and Technology, Tech. Rep.
  NIST Special Publication 800-207, Aug. 2020, accessed 2026-02-07. [Online].
  Available: \url{https://csrc.nist.gov/pubs/sp/800/207/final}
\BIBentrySTDinterwordspacing

\bibitem{Shor1994}
P.~W. Shor, ``Algorithms for quantum computation: Discrete logarithms and
  factoring,'' in \emph{Proceedings of the 35th Annual Symposium on Foundations
  of Computer Science}, 1994, pp. 124--134.

\bibitem{Grover1996}
L.~K. Grover, ``A fast quantum mechanical algorithm for database search,'' in
  \emph{Proceedings of the 28th Annual ACM Symposium on Theory of Computing},
  1996, pp. 212--219.

\bibitem{Sosnowski2023PQTLS13}
M.~Sosnowski, F.~Wiedner, E.~Hauser, L.~Steger, D.~Schoinianakis,
  S.~Gallenm{\"u}ller, and G.~Carle, ``The performance of post-quantum {TLS}
  1.3,'' in \emph{Companion of the 19th International Conference on Emerging
  Networking Experiments and Technologies}, 2023, pp. 19--27.

\bibitem{KampanakisChildsKlein2024}
P.~Kampanakis and W.~Childs-Klein, ``The impact of data-heavy, post-quantum
  {TLS} 1.3 on the time-to-last-byte of web connections,'' in \emph{Proceedings
  2024 Workshop on Measurements, Attacks, and Defenses for the Web}.\hskip 1em
  plus 0.5em minus 0.4em\relax Internet Society, 2024.

\bibitem{Sim2025KpqC}
\BIBentryALTinterwordspacing
M.~Sim, G.~Song, M.~Lee, S.~Yoon, A.~Baksi, and H.~Seo, ``Integrating and
  benchmarking {KpqC} in {TLS}/{X.509},'' 2025, accessed 2026-02-07. [Online].
  Available: \url{https://eprint.iacr.org/2025/1245}
\BIBentrySTDinterwordspacing

\bibitem{NISTIR8547ipd}
\BIBentryALTinterwordspacing
D.~Moody, R.~Perlner, A.~Regenscheid, A.~Robinson, and D.~Cooper, ``Transition
  to post-quantum cryptography standards,'' National Institute of Standards and
  Technology, Tech. Rep. NIST IR 8547 (Initial Public Draft), Nov. 2024,
  accessed 2026-02-07. [Online]. Available:
  \url{https://csrc.nist.gov/pubs/ir/8547/ipd}
\BIBentrySTDinterwordspacing

\bibitem{FIPS203}
\BIBentryALTinterwordspacing
{National Institute of Standards and Technology}, ``Module-lattice-based
  key-encapsulation mechanism standard,'' National Institute of Standards and
  Technology, Tech. Rep. FIPS 203, Aug. 2024, accessed 2026-02-07. [Online].
  Available: \url{https://csrc.nist.gov/pubs/fips/203/final}
\BIBentrySTDinterwordspacing

\bibitem{FIPS204}
\BIBentryALTinterwordspacing
------, ``Module-lattice-based digital signature standard,'' National Institute
  of Standards and Technology, Tech. Rep. FIPS 204, Aug. 2024, accessed
  2026-02-07. [Online]. Available:
  \url{https://csrc.nist.gov/pubs/fips/204/final}
\BIBentrySTDinterwordspacing

\bibitem{FIPS205}
\BIBentryALTinterwordspacing
------, ``Stateless hash-based digital signature standard,'' National Institute
  of Standards and Technology, Tech. Rep. FIPS 205, Aug. 2024, accessed
  2026-02-07. [Online]. Available:
  \url{https://csrc.nist.gov/pubs/fips/205/final}
\BIBentrySTDinterwordspacing

\bibitem{Google2017ALTS}
\BIBentryALTinterwordspacing
{Google Cloud Security and Privacy Team}, ``Securing communications between
  {Google} services with application layer transport security,'' Google Online
  Security Blog, Dec. 2017, accessed 2026-08-10. [Online]. Available:
  \url{https://security.googleblog.com/2017/12/securing-communications-between-google.html}
\BIBentrySTDinterwordspacing

\bibitem{BettaleDeOliveiraDottax2022}
\BIBentryALTinterwordspacing
L.~Bettale, M.~De~Oliveira, and E.~Dottax, ``Post-quantum protocols for banking
  applications,'' Fourth NIST PQC Standardization Conference, 2022, accessed
  2026-08-10. [Online]. Available:
  \url{https://csrc.nist.gov/csrc/media/Events/2022/fourth-pqc-standardization-conference/documents/papers/post-quantum-protocols-for-banking-applications-pqc2022.pdf}
\BIBentrySTDinterwordspacing

\bibitem{Google2026MTC}
\BIBentryALTinterwordspacing
{Chrome Secure Web and Networking Team}, ``Cultivating a robust and efficient
  quantum-safe {HTTPS},'' Google Online Security Blog, Feb. 2026, accessed
  2026-03-23. [Online]. Available:
  \url{https://security.googleblog.com/2026/02/cultivating-robust-and-efficient.html}
\BIBentrySTDinterwordspacing

\bibitem{Dahlhaus1992MultiwayCuts}
E.~Dahlhaus, D.~S. Johnson, C.~H. Papadimitriou, P.~D. Seymour, and
  M.~Yannakakis, ``The complexity of multiway cuts (extended abstract),'' in
  \emph{Proceedings of the Twenty-Fourth Annual ACM Symposium on Theory of
  Computing}, 1992, pp. 241--251.

\bibitem{KleinbergTardos2002}
J.~Kleinberg and {\'E}.~Tardos, ``Approximation algorithms for classification
  problems with pairwise relationships,'' \emph{Journal of the ACM}, vol.~49,
  no.~5, pp. 616--639, 2002.

\bibitem{GalloGrigoriadisTarjan1989}
G.~Gallo, M.~D. Grigoriadis, and R.~E. Tarjan, ``A fast parametric maximum flow
  algorithm and applications,'' \emph{SIAM Journal on Computing}, vol.~18,
  no.~1, pp. 30--55, 1989.

\bibitem{SPIREScaling}
\BIBentryALTinterwordspacing
{The SPIFFE Authors}, ``Scaling {SPIRE},'' {SPIFFE} deployment guidance, 2026,
  version 1.15.1, accessed 2026-07-17. [Online]. Available:
  \url{https://spiffe.io/docs/latest/planning/scaling_spire/}
\BIBentrySTDinterwordspacing

\bibitem{SPIFFEFederation}
\BIBentryALTinterwordspacing
------, ``{SPIFFE Federation},'' {SPIFFE} specification, 2026, version 1.15.1,
  accessed 2026-07-17. [Online]. Available:
  \url{https://spiffe.io/docs/latest/spiffe-specs/spiffe_federation/}
\BIBentrySTDinterwordspacing

\bibitem{RFC8693}
\BIBentryALTinterwordspacing
M.~Jones, A.~Nadalin, B.~Campbell, J.~Bradley, and C.~Mortimore, ``{OAuth 2.0
  Token Exchange},'' Internet Engineering Task Force (IETF), RFC 8693, Jan.
  2020. [Online]. Available: \url{https://www.rfc-editor.org/rfc/rfc8693}
\BIBentrySTDinterwordspacing

\bibitem{KarypisKumar1998Metis}
G.~Karypis and V.~Kumar, ``A fast and high quality multilevel scheme for
  partitioning irregular graphs,'' \emph{SIAM Journal on Scientific Computing},
  vol.~20, no.~1, pp. 359--392, 1998.

\bibitem{KernighanLin1970}
B.~W. Kernighan and S.~Lin, ``An efficient heuristic procedure for partitioning
  graphs,'' \emph{Bell System Technical Journal}, vol.~49, no.~2, pp. 291--307,
  1970.

\bibitem{FiducciaMattheyses1982}
C.~M. Fiduccia and R.~M. Mattheyses, ``A linear-time heuristic for improving
  network partitions,'' in \emph{Proceedings of the 19th Design Automation
  Conference}, 1982, pp. 175--181.

\bibitem{Steidl2022TrainTicketAnomalies}
M.~Steidl, ``Anomalies in microservice architecture (train-ticket) based on
  version configurations,'' Zenodo dataset, 2022.

\bibitem{PopeTauritzKent2019}
A.~S. Pope, D.~R. Tauritz, and A.~D. Kent, ``Evolving bipartite authentication
  graph partitions,'' \emph{IEEE Transactions on Dependable and Secure
  Computing}, vol.~16, no.~1, pp. 58--71, 2019.

\bibitem{KazdagliTiwariKumar2022}
\BIBentryALTinterwordspacing
M.~Kazdagli, M.~Tiwari, and A.~Kumar, ``Using constraint programming and graph
  representation learning for generating interpretable cloud security
  policies,'' arXiv:2205.01240, 2022. [Online]. Available:
  \url{https://arxiv.org/abs/2205.01240}
\BIBentrySTDinterwordspacing

\bibitem{WongGoudaLam1998KeyGraphs}
C.~K. Wong, M.~Gouda, and S.~S. Lam, ``Secure group communications using key
  graphs,'' \emph{ACM SIGCOMM Computer Communication Review}, vol.~28, no.~4,
  pp. 68--79, 1998.

\bibitem{WallnerHarderAgee1999RFC2627}
\BIBentryALTinterwordspacing
D.~Wallner, E.~Harder, and R.~Agee, ``Key management for multicast: Issues and
  architectures,'' RFC Editor, Tech. Rep. RFC 2627, Jun. 1999, accessed
  2026-02-07. [Online]. Available: \url{https://www.rfc-editor.org/rfc/rfc2627}
\BIBentrySTDinterwordspacing

\bibitem{ChanRajaramanSunZhu2009}
A.~Chan, R.~Rajaraman, Z.~Sun, and F.~Zhu, ``Approximation algorithms for key
  management in secure multicast,'' in \emph{Computing and Combinatorics}, ser.
  Lecture Notes in Computer Science.\hskip 1em plus 0.5em minus 0.4em\relax
  Springer, 2009, vol. 5609, pp. 148--157.

\bibitem{SakaiYamamoto2005}
\BIBentryALTinterwordspacing
H.~Sakai and H.~Yamamoto, ``Asymptotically optimal tree-based group key
  management schemes,'' arXiv:cs/0507001, 2005. [Online]. Available:
  \url{https://arxiv.org/abs/cs/0507001}
\BIBentrySTDinterwordspacing

\bibitem{Atallah2009KeyHierarchies}
M.~J. Atallah, M.~Blanton, N.~Fazio, and K.~B. Frikken, ``Dynamic and efficient
  key management for access hierarchies,'' \emph{ACM Transactions on
  Information and System Security}, vol.~12, no.~3, pp. 1--43, 2009, article
  13.

\bibitem{Birgisson2014Macaroons}
A.~Birgisson, J.~G. Politz, {\'U}.~Erlingsson, A.~Taly, M.~Vrable, and
  M.~Lentczner, ``Macaroons: Cookies with contextual caveats for decentralized
  authorization in the cloud,'' in \emph{Proceedings of the Network and
  Distributed System Security Symposium (NDSS)}.\hskip 1em plus 0.5em minus
  0.4em\relax Internet Society, 2014.

\bibitem{Andersen2019WAVE}
\BIBentryALTinterwordspacing
M.~P. Andersen, S.~Kumar, M.~AbdelBaky, G.~Fierro, J.~Kolb, H.-S. Kim, D.~E.
  Culler, and R.~A. Popa, ``{WAVE}: A decentralized authorization framework
  with transitive delegation,'' in \emph{28th USENIX Security Symposium (USENIX
  Security 19)}.\hskip 1em plus 0.5em minus 0.4em\relax Santa Clara, CA: USENIX
  Association, Aug. 2019, pp. 1375--1392. [Online]. Available:
  \url{https://www.usenix.org/conference/usenixsecurity19/presentation/andersen}
\BIBentrySTDinterwordspacing

\bibitem{WardBeyer2014BeyondCorp}
\BIBentryALTinterwordspacing
R.~Ward and B.~Beyer, ``{BeyondCorp}: A new approach to enterprise security,''
  \emph{;login:}, vol.~39, no.~6, pp. 6--11, 2014. [Online]. Available:
  \url{https://research.google/pubs/beyondcorp-a-new-approach-to-enterprise-security/}
\BIBentrySTDinterwordspacing

\bibitem{NoelSwarupJohnsgard2023Microsegmentation}
S.~Noel, V.~Swarup, and K.~Johnsgard, ``Optimizing network microsegmentation
  policy for cyber resilience,'' \emph{The Journal of Defense Modeling and
  Simulation: Applications, Methodology, Technology}, vol.~20, no.~1, pp.
  57--79, 2023.

\bibitem{Mani2025ZTS}
\BIBentryALTinterwordspacing
S.~K. Mani, K.~Hsieh, S.~Segarra, R.~Chandra, Y.~Zhou, and S.~Kandula,
  ``Securing public cloud networks with efficient role-based
  {Micro-Segmentation},'' in \emph{22nd USENIX Symposium on Networked Systems
  Design and Implementation (NSDI 25)}.\hskip 1em plus 0.5em minus 0.4em\relax
  Philadelphia, PA: USENIX Association, Apr. 2025, pp. 1033--1048. [Online].
  Available: \url{https://www.usenix.org/conference/nsdi25/presentation/mani}
\BIBentrySTDinterwordspacing

\bibitem{SkutellaWoeginger2000}
M.~Skutella and G.~J. Woeginger, ``A {PTAS} for minimizing the total weighted
  completion time on identical parallel machines,'' \emph{Mathematics of
  Operations Research}, vol.~25, no.~1, pp. 63--75, 2000.

\end{thebibliography}

\clearpage
\onecolumn
\setcounter{section}{0}
\setcounter{table}{0}
\setcounter{figure}{0}
\setcounter{equation}{0}
\renewcommand{\thesection}{S\arabic{section}}
\renewcommand{\thesubsection}{\thesection.\arabic{subsection}}
\renewcommand{\thesubsubsection}{\thesubsection.\arabic{subsubsection}}
\renewcommand{\thetable}{S\arabic{table}}
\renewcommand{\thefigure}{S\arabic{figure}}
\renewcommand{\theequation}{S\arabic{equation}}
\providecommand*{\theHtable}{}
\providecommand*{\theHfigure}{}
\renewcommand*{\theHsection}{SI.\arabic{section}}
\renewcommand*{\theHtable}{SI.\arabic{table}}
\renewcommand*{\theHfigure}{SI.\arabic{figure}}
\renewcommand*{\theHequation}{SI.\arabic{equation}}
\newcommand{\MainTableRef}[2]{Table~\ref{#1}}
\newcommand{\MainCorRef}[2]{Corollary~\ref{#1}}
\newcommand{\MainRemarkRef}[2]{Remark~\ref{#1}}
\newcommand{\MainSecRef}[2]{Sec.~\ref{#1}}
\newcommand{\MainEqRef}[2]{Eq.~(\ref{#1})}
\begin{center}
  {\Large\bfseries Supplementary Information:\\[0.25em]
  \PaperTitle\par}
  \vspace{0.75em}
  {\large \PaperAuthorOne{}\,\orcidlink{\PaperAuthorOneORCID} and \PaperAuthorTwo{}\,\orcidlink{\PaperAuthorTwoORCID}\par}
\end{center}
\vspace{0.75em}
\begingroup
  \renewcommand{\thefootnote}{}%
  \footnotetext{%
    \begin{list}{\textbullet}{%
      \setlength{\leftmargin}{1.2em}%
      \setlength{\labelwidth}{0.7em}%
      \setlength{\labelsep}{0.3em}%
      \setlength{\itemindent}{0pt}%
      \setlength{\listparindent}{0pt}%
      \setlength{\topsep}{0pt}%
      \setlength{\parsep}{0pt}%
      \setlength{\itemsep}{0pt}%
    }
      \item Corresponding author: \PaperCorrespondingAuthor{} (\href{mailto:\PaperCorrespondingEmail}{\PaperCorrespondingEmail}).
      \item \PaperAuthorOne{} and \PaperAuthorTwo{} are with \PaperAffiliationAddress.
    \end{list}%
  }%
\endgroup

The worked example makes trust-domain assignment and within-domain credential derivation concrete.
The remaining sections give structural results and proofs, define secondary risk summaries, report solver, workload, crossing-cost, and shared-issuer diagnostics, and describe deployment considerations.

\section{Worked example}\label{SIsec:worked_example}
We illustrate how feasibility constraints and intra-domain derivation structure affect blast radius and latency.
Throughout this example, the domain assignment $D$ is fixed and maps each service to one of two domains. The quantity $\mathrm{Lat}(D)$ is the total calibrated cost of interactions that cross between those domains. Fixing $D$ isolates derivation-side changes in blast radius at constant boundary latency.
Let $V=\{a,b,c,d,e,f\}$ and $k=2$ with $D(a)=D(b)=D(c)=1$ and $D(d)=D(e)=D(f)=2$.
Let the interaction edges include $a\to b$, $b\to c$, $c\to d$, $b\to e$, and $e\to f$.
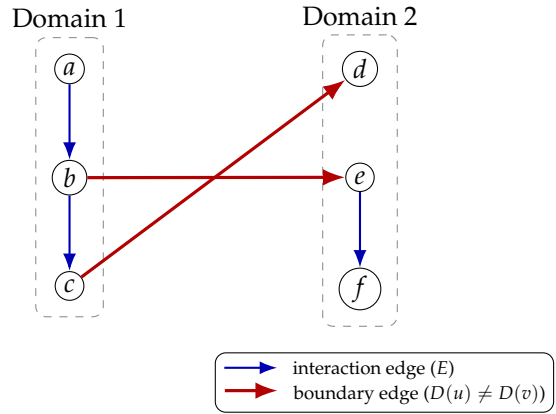
\begin{figure}[!htbp]
  \centering
  \begin{tikzpicture}[>=Latex]
  \node[svc] (a) {$a$};
  \node[svc,below=10mm of a] (b) {$b$};
  \node[svc,below=10mm of b] (c) {$c$};

  \node[svc,right=34mm of a] (d) {$d$};
  \node[svc,below=10mm of d] (e) {$e$};
  \node[svc,below=10mm of e] (f) {$f$};

  \node[domainbox,fit=(a)(b)(c),label=above:{Domain 1}] (box1) {};
  \node[domainbox,fit=(d)(e)(f),label=above:{Domain 2}] (box2) {};

  \draw[intedge] (a) -- (b);
  \draw[intedge] (b) -- (c);
  \draw[boundedge] (c) -- (d);
  \draw[boundedge] (b) -- (e);
  \draw[intedge] (e) -- (f);

  \node[
    anchor=west,
    draw,
    rounded corners,
    fill=white,
    inner sep=3pt,
    font=\scriptsize,
    align=left,
  ] (leg) at ($(box1.south)!0.5!(box2.south)+(0,-8mm)$) {%
    \tikz{\draw[intedge] (0,0) -- (7mm,0);} \; interaction edge ($E$)\\
    \tikz{\draw[boundedge] (0,0) -- (7mm,0);} \; boundary edge ($D(u)\ne D(v)$)%
  };
\end{tikzpicture}
  \caption{\label{SIfig:running_example}
  Fixed domains separate boundary latency from derivation risk.
  All arrows are service calls. Red arrows mark the two crossing calls, whose $r_{uv}\ell_{uv}c_{\mathrm{pqc}}(u,v)$ contributions sum to $\mathrm{Lat}(D)=220$ in the unit-normalized toy calibration.
  Changing only the derivation trees leaves this boundary latency fixed while allowing a different blast radius.}
\end{figure}
Assume only $c\to d$ and $b\to e$ cross the domain boundary.
Let $r_{uv}$ denote the interaction rate, $\ell_{uv}$ its path-sensitivity weight, and $c_{\mathrm{pqc}}(u,v)$ its effective crossing cost. Using a unit-normalized toy crossing cost on the crossing edges (not a literal classical or PQC deployment measurement), together with $r_{cd}=100$, $\ell_{cd}=2$, $c_{\mathrm{pqc}}(c,d)=1$ and $r_{be}=20$, $\ell_{be}=1$, $c_{\mathrm{pqc}}(b,e)=1$, we have
\begin{equation}
  \mathrm{Lat}(D) = 100\cdot 2\cdot 1 \;+\; 20\cdot 1\cdot 1 \;=\; 220.
\end{equation}
Let $H_i$ denote the credential-derivation tree rooted at issuer $\rho_i$, let $\Delta$ bound the number of children of any tree vertex, and let $h$ bound the root-to-service depth. In this example, $\mathcal A^{\mathrm{allow}}$ contains every root-to-service arc and only the service-to-service arcs $b\to c$ and $e\to f$.
With $\Delta=2$ and $h=2$, these constraints require $\rho_1\to a$, $\rho_1\to b$, and $b\to c$ in domain~1, and $\rho_2\to d$, $\rho_2\to e$, and $e\to f$ in domain~2.
Let $w(v)=1$ for all $v\in V$, and assume compromise probabilities $p(b)=0.02$ and $p(e)=0.01$ with all other $p(x)=0$.
Writing $\mathrm{Reach}_{H_i}(\{x\})$ for the services reachable from a compromised vertex $x$ in $H_i$, including $x$ itself, gives $\mathrm{Reach}_{H_1}(\{b\})=\{b,c\}$ and $\mathrm{Reach}_{H_2}(\{e\})=\{e,f\}$. The conservative score $\mathrm{BR}_{\mathrm{node}}$ sums each compromise probability times the total weight reached, so
\begin{equation}
  \mathrm{BR}_{\mathrm{node}}(D,\{H_i\})
  = 0.02\cdot 2 \;+\; 0.01\cdot 2 \;=\; 0.06.
\end{equation}
Increasing the degree bound to $\Delta=3$ permits each root to issue directly to all three services. Then $\mathrm{Reach}_{H_1}(\{b\})=\{b\}$ and $\mathrm{Reach}_{H_2}(\{e\})=\{e\}$, reducing the conservative blast-radius score to $0.02\cdot 1 + 0.01\cdot 1 = 0.03$ without changing $\mathrm{Lat}(D)$.
This demonstrates why derivation eligibility and fanout limits can change blast radius even when trust boundaries are fixed.
In this example the derivation structures are simple depth-two arborescences with the arcs listed above.

\subsection{Why latency-first optimization can fail}
Consider three unit-weight principals $a,b,c$, two nonempty domains, zero root risk, and chain derivations $(\Delta,h)=(1,2)$.
Let $p(a)=0.01$ and $p(b)=p(c)=0.9$.
The only interaction coefficients are $1$ on $a\to b$ and $2$ on $b\to c$.
For singleton $a$, $b$, or $c$, respectively, the best-chain triples $(\mathrm{Lat},\mathrm{BR}_{\mathrm{node}},\mathrm{BR}_{\mathrm{node}}+\tfrac12\mathrm{Lat})$ are $(1,2.71,3.21)$, $(3,1.82,3.32)$, and $(2,1.82,2.82)$.
A latency-first design therefore selects $\{a\}\mid\{b,c\}$, whereas joint scalarized optimization selects $\{a,b\}\mid\{c\}$.
Equivalently, under budget $B=2$, the latter reduces blast radius from $2.71$ to $1.82$.
The difference arises because grouping $b$ and $c$ forces one high-risk principal above the other, while grouping $a$ and $b$ places the low-risk principal first.

\FloatBarrier

\section{Formal results and proofs}\label{SIsup:proofs}
The formal results distinguish model-specific reductions from guarantees inherited from cited algorithms.
This section states domain-assignment and fixed-domain derivation results before their supporting proofs.
\subsection{Secondary structural results}\label{SIsup:secondary_results}
\paragraph{Direct-issuance assignment assumptions}
Every service must be eligible to receive credentials directly from any root, and root probabilities must remain fixed as services change domains. Domain-policy constraints are limited to the anchors specified in each result.

The degree and depth limits must permit direct issuance for every allowed assignment. We require $h\ge1$ and sufficient root fanout: $\Delta\ge |V|-k+1$ when all $k$ domains must be nonempty, or $\Delta\ge |V|$ when unused domains are permitted.

\begin{corollary}[Parametric cut sweep]\label{SIcor:parametric_cut_sweep}
Consider two nonempty domains with fixed label-specific root probabilities independent of $D$, where direct issuance is feasible and optimal. An anchor is a service fixed to a specified domain. Varying the latency weight $\lambda\in[0,\infty)$ in $\mathrm{BR}_{\mathrm{node}}+\lambda\mathrm{Lat}(D)$ traces the supported tradeoff between boundary latency and blast radius.
For each fixed ordered anchor pair, applying the parametric maximum-flow/min-cut algorithm of Gallo--Grigoriadis--Tarjan computes all breakpoints and a nested representative optimal cut for every parameter interval in polynomial time, rather than solving an independent cut on a dense parameter grid.\cite{GalloGrigoriadisTarjan1989}
Without fixed anchors, enumerating ordered anchor pairs and taking the lower envelope still gives a polynomial-time sweep. Representatives selected from different anchor pairs need not be globally nested.
\end{corollary}
To match the cited parametric-flow family, exchange labels so that $\delta=p(\rho_2)-p(\rho_1)\ge0$, subtract constants, divide the $\lambda>0$ objective by $\lambda$, and set $\mu=1/\lambda$.
The interaction capacities are then fixed and the source capacities $\mu\delta w(v)$ are monotone in $\mu$.
\begin{corollary}[Interval-robustness]\label{SIcor:interval_robustness}
Consider interval uncertainty in boundary-edge coefficients and compromise probabilities.
For each interaction edge $(u,v)\in E$, let $a_{uv}\in[\underline a_{uv},\overline a_{uv}]$, and for each compromise point $x$ (service or root) let $p(x)\in[\underline p(x),\overline p(x)]$.
Let $\overline a_{uv}$ and $\overline p(x)$ denote the upper endpoints, and write $\overline p$ for the collection $\{\overline p(x)\}_x$.
Then for any fixed design $(D,\{H_i\})$,
\[
\begin{aligned}
  \max_{\{a_{uv}\}}\mathrm{Lat}(D)
  &=
  \sum_{(u,v)\in E: D(u)\ne D(v)} \overline a_{uv}, \\
  \max_{\{p(x)\}}\mathrm{BR}_{\mathrm{node}}(D,\{H_i\})
  &=
  \mathrm{BR}_{\mathrm{node}}(D,\{H_i\})\big|_{p=\overline p}.
\end{aligned}
\]
Consequently, the robust counterparts of the scalarized objective (minimize worst-case scalarized objective) and the budgeted objective (minimize worst-case $\mathrm{BR}_{\mathrm{node}}$ subject to worst-case budget feasibility) reduce to the same problems with $a_{uv}$ and $p(x)$ replaced by their upper bounds.
In particular, under the assumptions of the two-domain min-cut regime, the interval-robust scalarized problem remains reducible to a minimum $s$--$t$ cut (using capacities $\overline a_{uv}$ and root-risk terms evaluated at $\overline p$), and the parametric sweep in Corollary~\ref{SIcor:parametric_cut_sweep} still applies.
\end{corollary}
The endpoint substitution follows directly because every uncertain coefficient has a nonnegative multiplier in both objectives.
\begin{corollary}[Multiway-cut regime]\label{SIcor:multiway_cut_regime}
Under these direct-issuance assumptions, with $k\ge3$, $\lambda>0$, and uniform root compromise probability $p(\rho_i)=p_0$ for all $i$, the $\mathrm{BR}_{\mathrm{node}}$ term becomes constant.
If, in addition, $k$ anchor vertices $t_1,\dots,t_k\in V$ are required to lie in distinct domains (a policy constraint in (C2)), then optimizing the scalarized objective over $D$ is equivalent to a minimum multiway cut instance on the underlying undirected interaction graph (ignoring directions, equivalently summing antiparallel weights) with terminals $\{t_1,\dots,t_k\}$.
This problem is NP-hard and admits a $(2-2/k)$-approximation algorithm.\cite{Dahlhaus1992MultiwayCuts}
\end{corollary}
At $\lambda=0$, all feasible assignments have the same scalarized objective value.
\begin{proposition}[Star regime: treewidth dynamic programming]\label{SIprop:star_treewidth}
Under the same star-derivation assumptions, for general $k\ge 2$ and arbitrary root compromise probabilities $p(\rho_i)$, the scalarized objective reduces (up to additive constants) to a uniform metric labeling objective, also called a Potts maximum-a-posteriori (MAP) objective.\cite{KleinbergTardos2002}
\[
  \min_{D:V\to[k]}\;
  \begin{aligned}[t]
    &\sum_{v\in V} c_{v,D(v)} \\
    &\quad+\;\sum_{(u,v)\in E} b_{uv}\,\mathbf{1}[D(u)\ne D(v)].
  \end{aligned}
\]
where $c_{v,i}=w(v)p(\rho_i)$ and $b_{uv}=\lambda\,r_{uv}\ell_{uv}c_{\mathrm{pqc}}(u,v)$.
Since $\mathbf{1}[D(u)\ne D(v)]$ is symmetric, one may equivalently view the edge-disagreement term as living on the underlying undirected graph with weights obtained by summing antiparallel directed interactions.
The displayed formulation permits unused labels and is therefore the at-most-$k$ relaxation of (C1).
If distinct anchors $t_1,\ldots,t_k$ are fixed to labels $1,\ldots,k$, the anchors enforce (C1). Given an $O(n)$-bag width-$\tau$ tree decomposition, dynamic programming (DP) then finds an exact optimum in $O(nk^{\tau+1})$ time.
Without fixed anchors, exact (C1) optimization follows by minimizing over all $(n)_k=n!/(n-k)!$ ordered anchor tuples, in $O((n)_k n k^{\tau+1})$ time and hence polynomial time for fixed $k$.
\end{proposition}
Proof in Supplementary Sec.~\ref{SIapp:proof_star_treewidth}.
\begin{corollary}[At-most-$k$ star-regime approximation]\label{SIcor:star_metric_labeling_approx}
The at-most-$k$ relaxation in Proposition~\ref{SIprop:star_treewidth} is a uniform metric labeling instance.
It admits a polynomial-time $2$-approximation algorithm.\cite{KleinbergTardos2002}
This inherited guarantee does not by itself apply to the exact-$k$ nonempty-domain constraint (C1).
\end{corollary}
\begin{corollary}[Treewidth under refinement]\label{SIcor:treewidth_refinement}
Let $G$ be the underlying undirected interaction graph and suppose $\mathrm{tw}(G)=\tau$.
Form a refined graph $G'$ by replacing each vertex $v\in V$ by between one and $s$ scoped principals (e.g., public vs.\ privileged).
Add edges only between refined endpoints of original edges.
The standard bag-expansion construction gives $\mathrm{tw}(G')\le s(\tau+1)-1$.
Under the direct-issuance assumptions of Proposition~\ref{SIprop:star_treewidth}, an expanded decomposition with $O(|V'|)$ bags gives runtime $O(|V'|\,k^{s(\tau+1)})$ for specified anchors or the at-most-$k$ relaxation.
For unanchored exact-$k$ optimization, enumerating ordered anchor tuples gives runtime $O((|V'|)_k\,|V'|\,k^{s(\tau+1)})$.
\end{corollary}
Proof in Supplementary Sec.~\ref{SIapp:proof_treewidth_refinement}.
\begin{remark}[Tightness]\label{SIrem:treewidth_refinement_tight}
The dependence on $s$ and $\tau$ is tight up to constants.
For $\tau\ge 1$, let $G$ be a clique on $\tau+1$ vertices and replace each vertex by an independent set of $s$ refined principals, with all cross-fiber edges induced by the original clique.
The resulting complete $(\tau+1)$-partite graph has equal part size $s$ and treewidth $s\tau$, within an additive $s-1$ of the upper bound $s(\tau+1)-1$.
\end{remark}

We now fix the domain assignment and optimize one derivation tree at a time.
For a fixed domain, write its contribution as
\begin{equation}
  \mathrm{BR}^{(i)}_{\mathrm{node}}(H_i)
  =\sum_{v\in V_i}w(v)\sum_{x\in\mathrm{Anc}_{H_i}(v)}p(x).
  \label{SIeq:br_domain}
\end{equation}

\begin{corollary}[Exact chain ordering ($\Delta=1$)]\label{SIcor:chain}
Fix a domain $V_i$ with weights $w(v)\ge 0$ and compromise probabilities $p(v)\ge 0$.
Assume (C3)--(C4) with $\Delta=1$, $h\ge |V_i|$, and complete derivation eligibility on $V_i\cup\{\rho_i\}$.
Any feasible $H_i$ is a directed chain rooted at $\rho_i$ and therefore induces an ordering (permutation) $\sigma$ of $V_i$.
Then a minimizer of the domain objective in Eq.~(\ref{SIeq:br_domain}) is obtained by sorting services by nondecreasing ratio $p(v)/w(v)$ (with the convention $p(v)/0=+\infty$ when $w(v)=0$).
This optimal chain can be computed in $O(|V_i|\log |V_i|)$ time.
\end{corollary}
Proof in Supplementary Sec.~\ref{SIapp:proof_chain}.

\begin{proposition}[Star optimality ($\Delta\ge |V_i|$)]\label{SIprop:star}
Fix a domain $V_i$ with weights $w(v)\ge 0$ and compromise probabilities $p(v)\ge 0$, and assume $\Delta\ge |V_i|$, $h\ge 1$, and eligibility of every arc $\rho_i\to v$ for $v\in V_i$.
Then the star arborescence with arcs $\rho_i\rightarrow v$ for all $v\in V_i$ minimizes $\mathrm{BR}^{(i)}_{\mathrm{node}}(H_i)$ among all feasible $H_i$.
\end{proposition}
Proof in Supplementary Sec.~\ref{SIapp:proof_star}.

\begin{proposition}[Depth-two case ($h=2$)]\label{SIprop:h2}
Fix a domain $V_i$ with $n=|V_i|$, assume complete derivation eligibility and $h=2$, and let $\Delta\ge 1$ (so feasibility requires $n\le \Delta+\Delta^2$).
If $n>\Delta$, there exists an optimal $H_i$ in which exactly $\Delta$ services are attached directly to $\rho_i$ (depth 1), and all remaining services have depth 2.
Fix any depth-1 set $U\subseteq V_i$ with $|U|=\Delta$ and define leaves $L=V_i\setminus U$.
An optimal depth-two arborescence consistent with hub set $U$ is obtained by solving the assignment problem below (up to additive constants in Eq.~(\ref{SIeq:br_domain})):
\begin{equation}
  \begin{aligned}
    \min_{\{x_{uv}\}} \quad & \sum_{u\in U}\sum_{v\in L} p(u)\,w(v)\,x_{uv} \\
    \text{s.t.}\quad & \sum_{u\in U} x_{uv} = 1 \quad \forall v\in L \\
    & \sum_{v\in L} x_{uv} \le \Delta \quad \forall u\in U \\
    & x_{uv}\in\{0,1\}.
  \end{aligned}
  \label{SIeq:h2_assignment}
\end{equation}
Because costs factor as $p(u)w(v)$, an optimal assignment attaches the largest $w(v)$ to the smallest $p(u)$ (fill the lowest-$p$ hub up to capacity $\Delta$, then proceed in increasing $p$).
If $\Delta$ is a fixed constant, enumerating all $\binom{n}{\Delta}$ hub sets and solving the corresponding assignment yields an exact algorithm running in $O(n^{\Delta}\,\mathrm{poly}(n))$ time.
\end{proposition}
Proof in Supplementary Sec.~\ref{SIapp:proof_h2}.

\subsection{Proof of coupled NP-hardness (chains)}\label{SIapp:proof_coupled_hard_chain}
\begin{proof}
We give a polynomial-time reduction from the scheduling problem $P_k\mathbin{\|}\sum_j w_j C_j$ on $k$ identical parallel machines.\footnote{In the standard three-field notation, $C_j$ denotes the completion time of job $j$ on its assigned machine.}
Skutella and Woeginger show that the problem is strongly NP-hard and admits a polynomial-time approximation scheme (PTAS).\cite{SkutellaWoeginger2000}

Fix $k\ge 3$ and take an instance with jobs $J=\{1,\dots,n\}$, processing times $t_j>0$, and weights $w_j\ge 0$.
Construct our instance as follows.
We may assume $n\ge k$: if $n<k$, an optimal schedule assigns each job to a distinct machine and can be found in polynomial time, so hardness is witnessed on instances with $n\ge k$.
Let $V$ contain one vertex $v_j$ per job $j$.
Set $E=\emptyset$ so that $\mathrm{Lat}(D)=0$ for all assignments.
Set the feasibility parameters $(\Delta,h)=(1,|V|)$.
Set all root compromise probabilities to zero: $p(\rho_i)=0$ for $i=1,\dots,k$.
Assign node compromise probabilities by scaling the processing times into $(0,1)$: let $T=1+\sum_{j\in J} t_j$ and define $p(v_j)=t_j/T$.
Set criticality weights $w(v_j)=w_j$.
This defines an instance satisfying the theorem's restrictions.

Because eligibility is complete, $\Delta=1$, and (C3) requires an arborescence, each feasible $H_i$ is a directed chain rooted at $\rho_i$ spanning $V_i$, i.e., any linear order of the vertices assigned to domain~$i$ is feasible.
Consider any feasible design $(D,\{H_i\})$ and interpret each vertex $v_j$ as job $j$, assigned to machine $D(v_j)$ and processed in the chain order induced by $H_{D(v_j)}$.
For a job $j$ scheduled on a machine with preceding jobs $j_1,\dots,j_m$, the completion time is
\[
  C_j \;=\; \sum_{\ell=1}^{m} t_{j_\ell} + t_j.
\]
In our objective, for the corresponding vertex $v_j$, the ancestor set along the chain (excluding the root, which has $p(\rho_i)=0$) is exactly $\{v_{j_1},\dots,v_{j_m},v_j\}$, so
\[
  \sum_{x\in \mathrm{Anc}_{H_{D(v_j)}}(v_j)} p(x) \;=\; \frac{1}{T}\Big(\sum_{\ell=1}^{m} t_{j_\ell} + t_j\Big) \;=\; \frac{C_j}{T}.
\]
Therefore, summing Eq.~(\ref{SIeq:br_domain}) over all domains and using $p(\rho_i)=0$,
\[
  \mathrm{BR}_{\mathrm{node}}(D,\{H_i\})
  \;=\;
  \sum_{j\in J} w_j \cdot \frac{C_j}{T}
  \;=\;
  \frac{1}{T}\sum_{j\in J} w_j C_j,
\]
so an optimizer of $\mathrm{BR}_{\mathrm{node}}$ yields an optimizer of $P_k\mathbin{\|}\sum_j w_j C_j$ (and vice versa), up to the positive scaling factor $1/T$.

\paragraph{Non-empty domains vs.\ unused machines}
Constraint (C1) requires each domain $V_i$ to be non-empty, corresponding to schedules that use all $k$ machines.
When $n\ge k$, this restriction does not change the scheduling optimum: given any schedule with an idle machine, select a machine with at least two jobs and move its last job to the idle machine.
This weakly decreases that job's completion time and leaves all other completion times unchanged, so $\sum_j w_j C_j$ does not increase.
By repeating, there exists an optimal schedule that uses all machines.
Therefore the scheduling optimum equals the optimum among schedules that use all machines, and our reduction is valid under (C1).
It follows that our coupled optimization is NP-hard for any fixed $k\ge 3$, and strongly NP-hard when $k$ is part of the input.
\end{proof}

\subsection{Proof of the two-domain min-cut regime}\label{SIapp:proof_coupled_min_cut_star}
\begin{proof}
Let $k=2$.
Under the assumptions of the proposition, for every assignment satisfying (C1) we may restrict attention to star derivations within each domain (Proposition~\ref{SIprop:star}).
In a star, each service $v$ has ancestors $\mathrm{Anc}(v)=\{\rho_{D(v)},v\}$, so
\[
  \mathrm{BR}_{\mathrm{node}}(D,\{H_i^{\star}\})
  \;=\;
  \sum_{v\in V} w(v)p(v)
  \;+\;
  \sum_{v\in V} w(v)\,p(\rho_{D(v)}),
\]
where the first term is constant over all $D$.
Let $a_{uv}=r_{uv}\ell_{uv}c_{\mathrm{pqc}}(u,v)$ so that $\mathrm{Lat}(D)=\sum_{(u,v)\in E} a_{uv}\,\mathbf{1}[D(u)\ne D(v)]$.
Dropping additive constants, the scalarized objective reduces to minimizing over $D$ the objective
\[
  \sum_{v\in V} c_{D(v)}(v) \;+\; \sum_{(u,v)\in E} b_{uv}\,\mathbf{1}[D(u)\ne D(v)],
\]
where $c_i(v)=p(\rho_i)w(v)$ and $b_{uv}=\lambda a_{uv}$.

Build a directed $s$--$t$ network with vertices $\{s,t\}\cup V$.
For each $v\in V$, add arcs $s\to v$ of capacity $c_2(v)$ and $v\to t$ of capacity $c_1(v)$.
For each interaction edge $(u,v)\in E$, add arcs $u\to v$ and $v\to u$, each of capacity $b_{uv}$.
For any $s$--$t$ cut $(S,T)$ (with $s\in S$, $t\in T$), define a partition $D$ by assigning $D(v)=1$ if $v\in S$ and $D(v)=2$ if $v\in T$.
Then the cut capacity equals the objective above: a vertex $v$ contributes $c_1(v)$ iff $v\in S$ (via $v\to t$ crossing), and $c_2(v)$ iff $v\in T$ (via $s\to v$ crossing).
For each $(u,v)\in E$, neither symmetric arc crosses when $D(u)=D(v)$. When $D(u)\ne D(v)$, exactly one crosses and contributes $b_{uv}$.
Thus, minimizing the scalarized objective is equivalent to a minimum $s$--$t$ cut and is solvable in polynomial time.
\paragraph{Enforcing non-empty domains}
Constraint (C1) requires that both domains contain at least one service vertex.
If two anchors $a,b\in V$ are specified and must satisfy $D(a)=1$ and $D(b)=2$, enforce these constraints by adding arcs $s\to a$ and $b\to t$ with capacity larger than any finite feasible cut (equivalently, treat $a$ and $b$ as fixed terminals on the $s$ and $t$ sides).
If no anchors are specified, enforce non-emptiness by taking the minimum over all ordered pairs $(a,b)$ of distinct vertices: for each pair, solve the anchored min-cut instance above and return the best solution.
This adds a factor of $O(|V|^2)$ to the runtime and remains polynomial.
\end{proof}

\subsection{Derivation of Corollary~\ref{SIcor:chain}}\label{SIapp:proof_chain}
\begin{proof}
Complete derivation eligibility and constraints (C3)--(C4) with $\Delta=1$ make every ordering of $V_i$ a feasible chain.
For chain order $(v_1,\ldots,v_n)$, its objective is
\[
  p(\rho_i)\sum_{t=1}^{n}w(v_t)
  +\sum_{t=1}^{n}w(v_t)\sum_{j=1}^{t}p(v_j).
\]
The first term is constant, and the second is $\sum_jw_jC_j$ with processing times $p(v_j)$.
For adjacent services $u$ and $v$, placing $u$ before $v$ rather than $v$ before $u$ changes this term by $w(v)p(u)-w(u)p(v)$.
Thus $u$ before $v$ is no worse whenever $p(u)/w(u)\le p(v)/w(v)$, using the convention in the corollary.
Repeatedly removing inversions proves that sorting by this ratio is optimal and takes $O(n\log n)$ time.
\end{proof}

\subsection{Proof of the additive issuer-score upper bound}\label{SIapp:proof_explicit_compiled_upper_bound}
\begin{proof}
Fix a domain $i$ and an issuer $a$, and abbreviate $T=T_i(a)$.
For each accepted target $v\in T$, let
\[
  \operatorname{Desc}_{H_i}(v)=\{u\in V_i : v\leadsto u \text{ in } H_i\}
\]
denote the descendants unlocked from $v$.
The explicit issuer contribution in domain $i$ is therefore
\[
  p(a)\sum_{u\in \bigcup_{v\in T} \operatorname{Desc}_{H_i}(v)} w(u)
  \;=\;
  p(a)\,W_{H_i}(T).
\]
The additive score instead counts each accepted target separately and contributes
\[
  p(a)\sum_{v\in T}\sum_{u\in \operatorname{Desc}_{H_i}(v)} w(u)
  \;=\;
  p(a)\sum_{v\in T} W_{H_i}(\{v\}).
\]
Because $w\ge 0$, the weight of a union is at most the sum of the individual weights:
\[
  \sum_{u\in \bigcup_{v\in T} \operatorname{Desc}_{H_i}(v)} w(u)
  \;\le\;
  \sum_{v\in T}\sum_{u\in \operatorname{Desc}_{H_i}(v)} w(u).
\]
Multiplying by $p(a)\ge 0$ preserves the inequality, so the additive contribution of issuer $a$ in domain $i$ is an upper bound on the explicit contribution.
Summing over all issuers and all domains, and then adding back the unchanged service-compromise and root terms, proves the pointwise upper bound on the full first-order objective.

Now assume that no accepted target in $T_i(a)$ is an ancestor of another.
In an arborescence, descendant sets of incomparable vertices are disjoint.
Hence the sets $\{\operatorname{Desc}_{H_i}(v)\}_{v\in T}$ are pairwise disjoint, so the union-weight inequality above is tight.
Therefore the additive and explicit contributions agree exactly issuer-by-issuer and domain-by-domain, and thus the full scores are equal.
In a star, each service has only itself as a service descendant, so accepted targets always form an antichain.
In a chain, all distinct services are comparable, so the antichain condition permits at most one target per issuer and domain.
For a target at position $j$, explicit propagation and the corresponding increase to $p_{\mathrm{eff}}$ both contribute $p(a)\sum_{t=j}^{m}w(v_t)$.
These observations give the star- and chain-overlay corollaries.
\end{proof}

\subsection{Proof of Proposition~\ref{SIprop:star}}\label{SIapp:proof_star}
\begin{proof}
Because $\Delta\ge |V_i|$, $h\ge 1$, and every root-to-service arc is eligible, the star with arcs $\rho_i\to v$ for all $v\in V_i$ is feasible under (C3)--(C4).
In any feasible $H_i$, write Eq.~(\ref{SIeq:br_domain}) as
\[
  \begin{aligned}
    \mathrm{BR}^{(i)}_{\mathrm{node}}(H_i)
    &=
    p(\rho_i)\sum_{v\in V_i} w(v)
    \;+\;
    \sum_{v\in V_i} w(v)p(v) \\
    &\quad+\;
    \sum_{v\in V_i}\sum_{u\in \mathrm{Anc}_{H_i}(v)\setminus\{\rho_i,v\}} w(v)p(u).
  \end{aligned}
\]
The last double sum is nonnegative (since $w,p\ge 0$) and is identically zero for the star, because $\mathrm{Anc}(v)=\{\rho_i,v\}$ in a star.
Therefore, for any feasible $H_i$, $\mathrm{BR}^{(i)}_{\mathrm{node}}(H_i)\ge \mathrm{BR}^{(i)}_{\mathrm{node}}(H_i^{\star})$, proving optimality.
\end{proof}

\subsection{Proof of Proposition~\ref{SIprop:h2}}\label{SIapp:proof_h2}
\begin{proof}
Let $n=|V_i|$ and assume $h=2$.
For any feasible $H_i$, define the depth-1 set (children of the root)
\[
  U=\{u\in V_i : (\rho_i\to u)\in A_i\}.
\]
By the out-degree bound in (C4), $|U|\le \Delta$.
Every $v\in V_i\setminus U$ has depth two and therefore has a parent in $U$.

\emph{Step 1 (saturating the root degree).}
Assume $n>\Delta$ and take any feasible $H_i$ with $|U|<\Delta$.
Pick any leaf $v\in V_i\setminus U$ and let $u\in U$ be its parent.
Form $H_i'$ by replacing the arc $u\to v$ with $\rho_i\to v$.
This preserves indegree one for all vertices, increases $\deg^+(\rho_i)$ by one (still $\le\Delta$), decreases $\deg^+(u)$ by one, and does not increase depth.
Thus $H_i'$ is feasible.
Only the ancestor set of $v$ changes: $u$ is removed from $\mathrm{Anc}(v)$, while all other services have the same ancestors.
Therefore, the objective decreases by exactly $p(u)w(v)\ge 0$.
Repeating this promotion until $|U|=\Delta$ shows that some optimum satisfies $|U|=\Delta$ and all remaining nodes have depth two.

\emph{Step 2 (assignment formulation for fixed hubs).}
Fix a hub set $U\subseteq V_i$ with $|U|=\Delta$ and define $L=V_i\setminus U$.
In any depth-two arborescence consistent with $U$, each leaf $v\in L$ chooses a parent $u\in U$ and each $u$ may have at most $\Delta$ children.
Let $x_{uv}\in\{0,1\}$ indicate whether $v$ attaches to $u$.
Then each $v\in L$ has exactly one parent ($\sum_{u\in U}x_{uv}=1$) and each hub satisfies the degree constraint ($\sum_{v\in L}x_{uv}\le\Delta$).
For such a structure, hub nodes have ancestors $\{\rho_i,u\}$, and leaves have ancestors $\{\rho_i,u,v\}$, yielding
\[
  \begin{aligned}
    \mathrm{BR}^{(i)}_{\mathrm{node}}(H_i)
    &=
    p(\rho_i)\sum_{v\in V_i} w(v)
    \;+\;
    \sum_{v\in V_i} w(v)p(v) \\
    &\quad+\;
    \sum_{u\in U}\sum_{v\in L} p(u)w(v)x_{uv}.
  \end{aligned}
\]
The first two terms are constant given $V_i$, so minimizing $\mathrm{BR}^{(i)}_{\mathrm{node}}$ over depth-two structures with hub set $U$ is equivalent to the assignment problem in Eq.~(\ref{SIeq:h2_assignment}).

\emph{Step 3 (greedy optimality).}
Order hubs so that $p(u_1)\le \cdots \le p(u_\Delta)$ and order leaves so that $w(v_1)\ge \cdots \ge w(v_{n-\Delta})$.
We first observe that some optimal solution fills hubs in this order: if some $u_a$ has remaining capacity while a leaf $v$ is assigned to $u_b$ with $p(u_b)>p(u_a)$, moving $v$ from $u_b$ to $u_a$ does not increase cost, and decreases it when $w(v)>0$.
Given this, consider any feasible assignment and any two leaves $v,v'$ assigned to hubs $u_a,u_b$ with $p(u_a)\le p(u_b)$ while $w(v)<w(v')$.
Swapping the parents of $v$ and $v'$ preserves feasibility (each hub keeps the same number of children) and changes the objective by
\[
  \begin{aligned}
    &\big(p(u_a)w(v')+p(u_b)w(v)\big) - \big(p(u_a)w(v)+p(u_b)w(v')\big) \\
    &\qquad=\;
    (p(u_b)-p(u_a))\,(w(v)-w(v')) \;\le\; 0.
  \end{aligned}
\]
Thus repeated exchanges yield an optimal assignment in which larger weights are assigned to smaller-$p$ hubs, which is exactly the stated greedy rule (fill $u_1$ up to capacity $\Delta$ with the largest leaves, then $u_2$, and so on).

\emph{Step 4 (exactness for constant $\Delta$).}
If $\Delta$ is constant, enumerating all $\binom{n}{\Delta}=O(n^\Delta)$ hub sets and computing the optimal assignment for each (e.g., by the greedy rule above, or by min-cost flow) yields an exact algorithm in $O(n^\Delta\,\mathrm{poly}(n))$ time.
\end{proof}

\subsection{Proof of Proposition~\ref{SIprop:star_treewidth}}\label{SIapp:proof_star_treewidth}
\begin{proof}
In a star, each service contributes the assignment-dependent unary term $w(v)p(\rho_{D(v)})$. The latency term supplies the pairwise disagreement costs displayed in the proposition.
This is the stated Potts objective.
Under the stated policy assumptions, only the specified anchors impose hard constraints on the labels.
For a fixed anchored instance, let each bag table contain one entry for every assignment of its vertices to $[k]$.
An entry stores the minimum cost in the processed subgraph conditional on that bag assignment, with each unary and pairwise term charged when its last required vertex is processed.
Introduce, forget, and join transitions preserve this invariant, while anchor violations receive infinite cost.
A bag contains at most $\tau+1$ vertices, so it has at most $k^{\tau+1}$ entries, and each transition takes $O(k^{\tau+1})$ time.
The $O(n)$ bags therefore give total time $O(nk^{\tau+1})$.
Hard unary constraints fix anchor $t_i$ to label $i$, so every feasible assignment uses all $k$ labels.
Conversely, every assignment satisfying (C1) contains at least one ordered tuple of representatives $(t_1,\ldots,t_k)$ with $D(t_i)=i$.
Taking the minimum over all $(n)_k$ tuples therefore recovers the exact (C1) optimum and gives the stated runtime.
\end{proof}

\subsection{Proof of Corollary~\ref{SIcor:treewidth_refinement}}\label{SIapp:proof_treewidth_refinement}
\begin{proof}
Take any tree decomposition of the underlying undirected graph $G$ with bags $\{B_t\}$ of size at most $\tau+1$.
For each original vertex $v$, let $S_v$ denote its refined principal set with $|S_v|\le s$.
Replace each bag by
\[
  B'_t = \bigcup_{v\in B_t} S_v.
\]
Every refined edge lies within some bag $B'_t$ because its endpoints come from the refined endpoints of an original edge whose endpoints co-occur in some bag of the original decomposition.
The connectedness condition for each refined vertex follows from the connectedness condition for its parent vertex $v$.
Hence $\{B'_t\}$ is a valid tree decomposition of the refined graph $G'$.
Each refined bag has size at most $s(\tau+1)$, so $\mathrm{tw}(G')\le s(\tau+1)-1$.
An expanded decomposition with $O(|V'|)$ bags gives $O(|V'|\,k^{s(\tau+1)})$ time for specified anchors or the at-most-$k$ relaxation by Proposition~\ref{SIprop:star_treewidth}.
For unanchored exact-$k$ optimization, minimizing over the $(|V'|)_k$ ordered anchor tuples gives the additional factor.
\end{proof}

\section{Alternative risk summaries}\label{SIsup:risk_summaries}
For any set $X$ of compromise points, define its deterministic total impact by
\[
\mathrm{BR}(X,D,\{H_i\})
=\sum_{i=1}^{k}\sum_{v\in
\mathrm{Reach}_{H_i}(X\cap(V_i\cup\{\rho_i\}))}w(v).
\]
The worst-case single-compromise impact is
\[
\mathrm{BR}_{\max}(D,\{H_i\})=\max_x \mathrm{BR}(\{x\},D,\{H_i\}),
\]
When only domain-root compromise is retained, the issuer summary is
\[
\mathrm{BR}_{\mathrm{issuer}}(D)
=\sum_i p(\rho_i)\sum_{v:D(v)=i}w(v).
\]
If root risks are equal, this expectation is constant in $D$. Segmentation then appears in worst-case or heterogeneous-risk summaries.

\section{Supporting evaluation diagnostics}\label{SIsup:eval_tables}
This supplement reports experimental settings, calibration data, solver-quality diagnostics, workload variants, and explicit-propagation checks supporting the evaluation.

\subsection{Evaluation protocol}\label{SIsup:eval_protocol}
Candidate designs come from local search over normalized risk--latency tradeoffs and from deterministic traffic partitions. Each study states its scalarization grid. The principal trace replay uses nine equally spaced $\alpha$ values on $[0,1]$. The reference scales are the all-crossing latency $L_0$ and total impact $R_0$.
For each latency budget, budget-constrained moves and swaps refine every feasible candidate before the minimum-blast-radius design is selected.
All runs use fixed computational budgets.
Replay experiments assume complete within-domain derivation eligibility, use $(\Delta,h)=(3,3)$, and visit only nonempty partitions that satisfy the domain-size capacity bound implied by (C3)--(C4).
Within each domain, services are ordered by $p(v)/w(v)$ and attached in that order to the earliest parent with remaining fanout, proceeding breadth first to depth $h$.
Each partition is therefore scored with one deterministic feasible derivation family. Optimization over all admissible intra-domain arborescences lies outside these experiments.
The replay experiments set $\ell_{uv}=1$, enumerate $k\in\{1,\ldots,6\}$, set all domain-root compromise priors to zero, and apply no fixed per-domain management charge.
The upper limit on $k$ is a common computational planning cap, not a value inferred from latency or the derivation-capacity bound. A row selecting $k=6$ is therefore right-censored at the largest tested count.

Across the replay text and tables, $\mathrm{BR}$ abbreviates the conservative $\mathrm{BR}_{\mathrm{node}}$ score, and $\mathrm{BR}_0$ denotes its $k=1$ value under the same workload, risk setting, derivation method, and feasibility filter.
Synthetic selections are generated from fixed seeds and pinned NumPy dependencies.
The generation scripts record a SHA-256 manifest for the resulting outputs.
Reported wall-clock measurements were collected on \RuntimeMachine{} using Python \RuntimePython{} and NumPy \RuntimeNumPy{}.
Each study states its specific settings alongside its results.

\subsection{Q1: Solver validity and transfer}
\suppressfloats[t]
The exact coupling ablation in \MainTableRef{tab:joint_vs_staged}{the coupling table} fixes $n=9$ and $k=3$, then evaluates 20 clustered-graph seeds under the heterogeneous-baseline (H) and cluster-skew (S) risk settings at latency budgets $B\in\{0.20,0.30,0.40\}$\,ms.
For each feasible assignment, the reference computes both the direct-issuance score and the chain or depth-two score.
The staged baseline minimizes the direct-issuance score under $B$ and then selects the smallest chain or depth-two score among all tied optima.
Each seed assigns the three domain labels deterministic heterogeneous root priors drawn uniformly from $[0.01,0.07)$. These priors remain fixed across assignments, budgets, and the H and S risk settings.
This favorable tie rule isolates the loss caused by choosing boundaries under the simpler issuance structure.

The two-domain transfer test uses clustered synthetic graphs with $n\in\{20,40\}$ and seeds $0$--$5$.
Each instance is evaluated at \FormalGridCalls{} equally spaced $\alpha$ values. The grid method solves every value independently.
The adaptive supported-point sweep starts from $\alpha=0$ and $1$, then adds exact solves only where the current solutions imply another supported tradeoff. Both methods solve the same scalarized two-domain direct-issuance problem with nonempty domains and no fixed anchors.
The recursion may omit tied designs on a collinear supported segment, so the comparison checks objective equality at the grid values rather than recovery of every tied design.
Across the 12 instances, the sweep matched every one of the \FormalImpactChecks{} grid objective values.
Depending on $n$, it required \FormalSweepCallsLow--\FormalSweepCallsHigh{} exact solves per instance on average, compared with \FormalGridCalls{} independent solves for the grid. This solver-call comparison leaves runtime unmeasured and does not cover unsupported hard-budget optima.

The scaling study uses a clustered synthetic family with $n\in\{50,100,200,500,1000\}$, expected average out-degree 12, $k\le6$, four restarts, $\alpha\in\{0,0.25,0.5,0.75,1\}$, at most 250 iterations, and $(\Delta,h)=(6,3)$.
\begin{figure}[tbp]
  \centering
  \begin{tikzpicture}
\begin{axis}[
  width=0.98\linewidth,
  height=6.2cm,
  xmode=log,
  ymode=log,
  log basis x=10,
  log basis y=10,
  xlabel={Services $n$},
  ylabel={Mean frontier-search runtime (s)},
  grid=both,
  grid style={line width=.1pt, draw=gray!30},
  major grid style={line width=.2pt,draw=gray!45},
]
\addplot+[thick,mark=*,mark size=1.4pt,color=black,error bars/.cd,y dir=both,y explicit] table [x=n,y=mean_s,y error minus=err_minus_s,y error plus=err_plus_s,col sep=comma] {data/fig_inputs/solver_scaling_summary.csv};
\end{axis}
\end{tikzpicture}
  \caption{\label{SIfig:solver_scaling}
  Fixed-schedule frontier search remains within the tens-of-seconds range on the tested sparse family.
  Increasing $n$ from 50 to 1000 raises mean runtime from $\ScalingRuntimeFifty$ to $\ScalingRuntimeThousand$ seconds, a $\ScalingRuntimeFactor\times$ increase for $20\times$ more services.
  Points are means over two seeded instances per size, and bars span the observed minimum and maximum.
  Fixed restart, scalarization, and iteration limits cap the work, so the curve measures throughput for this schedule rather than worst-case scaling or solution quality.}
\end{figure}

The scaling-quality check compares the reduced schedule with exhaustive optimization over partitions with chain derivations, for $n\in\{8,9,10\}$ and $1\le k\le6$. It matches all \SolverQualityCases{} small-instance budget optima across five seeds and four latency budgets per size.

\begin{table}[tbp]
  \caption{\label{SItab:solver_budget_path_ablation}
  Across \BudgetPathCases{} chain-instance and budget combinations with $1\le k\le6$, budget repair lowers the miss rate, and adding traffic-based starts eliminates misses in all tested cases.
  Supported points provide the exact scalarization-supported reference. The remaining rows show reduced-schedule scalarized search, then add budget-feasible moves and swaps, followed by deterministic traffic-based starts with repair.
  Miss rate and gap are measured against exhaustive budget optima.}
  \centering
  \small
  \PaperTableRows
  \PaperTableRowsSetup
\begin{tabular}{lrrr}
\PaperTableHideRows
\toprule
Method & Miss (\%) & Mean gap (\%) & Max gap (\%) \\
\midrule
\PaperTableBodyRows
Supported points & 23.3 & 1.44 & 18.44 \\
Scalarized search & 51.7 & 3.74 & 20.31 \\
Budget repair & 3.3 & 0.04 & 1.90 \\
Traffic starts + repair & 0.0 & 0.00 & 0.00 \\
\hiderowcolors
\bottomrule
\end{tabular}

\end{table}

In a separate exact-chain check on clustered $n=9$, $k\le4$ instances (seeds 0--4), exhaustive enumeration of nonempty partitions and exact chain ordering provides the reference for the budget-refined search.
The heuristic and exact solutions choose the same $k$ in both observed miss cases.
The remaining failure mode is within-$k$ partition quality under tight budgets, where local search can still leave more budget unused than the exact solution (at $B=0.10$, mean slack $\SmallGapHeuristicSlack$ vs.\ $\SmallGapExactSlack$\,ms among misses).
The miss at $B=\SmallGapNearTieBudget$ is a near-tie with a very small objective gap.

\begin{table}[tbp]
  \caption{\label{SItab:br_exact_regret_stress}
  For chain derivations under the tested independent-event model, optimizing $\mathrm{BR}_{\mathrm{node}}$ produces zero or small $\mathrm{BR}_{\mathrm{exact}}$ regret.
  Across eight $n=9$ instances, with exhaustive partition search over $1\le k\le4$ and exact optimization of chain order for each objective, the largest observed regret is $\BRExactMaxRegret\%$.
  H denotes the heterogeneous-baseline risk scenario. S additionally triples service priors in the designated cluster. In the second block, every service prior is first tripled. All scaling is clipped at $0.25$.
  Entries report mean $\pm$ normal-approximation 95\% confidence half-width and maximum regret. $B$ is in ms/request.}
  \centering
  \small
  \PaperTableRowsSetup
\begin{tabular}{crrrr}
\PaperTableHideRows
\toprule
$B$ (ms/request) & \multicolumn{2}{c}{H scenario} & \multicolumn{2}{c}{S scenario} \\
\cmidrule(lr){2-3}\cmidrule(l){4-5}
 & Mean regret (\%) & Maximum (\%) & Mean regret (\%) & Maximum (\%) \\
\midrule
\PaperTableBodyRows
\multicolumn{5}{l}{\textbf{Baseline service priors}} \\
0.10 & 0.00$\pm$0.00 & 0.00 & 0.15$\pm$0.30 & 1.22 \\
0.20 & 0.00$\pm$0.00 & 0.00 & 0.00$\pm$0.00 & 0.00 \\
0.40 & 0.00$\pm$0.00 & 0.00 & 0.03$\pm$0.03 & 0.14 \\
\midrule
\multicolumn{5}{l}{\textbf{All service priors $\times 3.0$ (clipped at $0.25$)}} \\
0.10 & 0.00$\pm$0.00 & 0.00 & 0.12$\pm$0.16 & 0.59 \\
0.20 & 0.00$\pm$0.00 & 0.00 & 0.01$\pm$0.02 & 0.09 \\
0.40 & 0.01$\pm$0.02 & 0.07 & 0.00$\pm$0.00 & 0.00 \\
\hiderowcolors
\bottomrule
\end{tabular}

\end{table}

The results in Table~\ref{SItab:br_exact_regret_stress} support the surrogate over the tested probability ranges under independence.

Topology alone does not establish semantic eligibility for a formal solver.
Central issuance maps to a star, one intermediate layer to depth two, and sequential delegation to a chain. Shared verifier acceptance instead requires explicit propagation.
SPIRE documents single-authority, nested-intermediate, and federated layouts, while OAuth 2.0 Token Exchange represents impersonation and delegation chains.\cite{SPIREScaling,SPIFFEFederation,RFC8693}
The two-domain cut additionally requires a binary split, and fixed-anchor results require policy-pinned services.
Issuer, trust-bundle, certificate-chain, delegation, and verifier metadata can recover these conditions. Service-call traces cannot.
Thus, the \ReplayWorkloads{} replay graphs establish interaction-graph widths up to \ReplayMaxWidth, and the standards establish realizability. Deployment prevalence remains unmeasured.
The chain regime provides an exact best response and a hardness boundary. It is not treated as a default architecture.
Across the \ReplayHeterogeneousCases{} heterogeneous-root cases, anchored $\alpha$-expansion and anchored local search match the exact-DP-seeded budget selection.

\begin{table}[tbp]
  \caption{\label{SItab:regime_applicability}
  The min-fill heuristic produces tree decompositions of width at most \ReplayMaxWidth{} for the tested interaction graphs.
  For anchored direct issuance under the configured candidate-generation schedule, exact dynamic programming is faster on five workloads. On Train-Ticket, generating the candidate pool with anchored $\alpha$-expansion takes \ReplayTrainAlphaDPRatio$\times$ the exact-DP time.
  The final column reports this ratio before both methods undergo the same budget-refinement step. It measures the configured schedule rather than isolated solver complexity.}
  \centering
  \small
  \setlength{\tabcolsep}{4pt}
  \PaperTableRowsSetup
\begin{tabularx}{\linewidth}{>{\raggedright\arraybackslash}Xrr}
\PaperTableHideRows
\toprule
Workload & Min-fill width & $\alpha$-expansion / exact-DP time \\
\midrule
\PaperTableBodyRows
Fintech proxy & 3 & 1.46$\times$ \\
Online Boutique & 2 & 3.42$\times$ \\
Retail proxy & 3 & 1.98$\times$ \\
Sample architecture & 4 & 1.06$\times$ \\
socialNetwork & 2 & 2.73$\times$ \\
Train-Ticket & 6 & 0.09$\times$ \\
\hiderowcolors
\bottomrule
\end{tabularx}

\end{table}

The homogeneous-root control has zero blast-radius gap for every method because blast radius is partition-invariant under direct issuance.

The joint-chain transfer uses six seeded $n=8$, $k=3$ clustered instances, four fixed-anchor starts per seed, heterogeneous concrete root priors, and $\alpha\in\{0.2,0.4,0.6,0.8\}$.
The exact reference enumerates all labeled assignments consistent with the anchors and applies exact ratio ordering within every domain.
\begin{table}[tbp]
  \caption{\label{SItab:joint_chain_transfer}
  Local search recovers the exact assignment in \ChainLocalSameAssignment\% of the tested chain cases and reaches a maximum objective gap of \ChainLocalMaxGap\%, despite its lower measured runtime.
  The exact method enumerates anchor-consistent assignments and applies exact ratio ordering within each domain.
  The experiment tests an eligible chain family and does not estimate the prevalence of chain delegation.}
  \centering
  \small
  \PaperTableRows
  \PaperTableRowsSetup
\begin{tabular}{lrrrrr}
\PaperTableHideRows
\toprule
Method & Cases & Mean gap (\%) & Max gap (\%) & Exact assignment (\%) & Mean time (ms) \\
\midrule
\PaperTableBodyRows
Exact & 96 & 0.00 & 0.00 & 100 & 2.61 \\
Local & 96 & 9.36 & 84.30 & 43 & 0.07 \\
\hiderowcolors
\bottomrule
\end{tabular}

\end{table}

The refinement stress test uses scope sizes $2$--$4$ and width caps $4,6,8$.
\begin{table}[!htbp]
  \caption{\label{SItab:formal_refinement_stress}
  Under anchored direct issuance, scoped refinement increases min-fill width and exact bounded-treewidth solver runtime in the tested families.
  Larger scopes require higher width caps. At scope size four, only the tree family is solved within the tested width caps.
  Width and runtime ratios are means of refined-to-base ratios across the tested cases. The last two columns report the smallest tested cap that solves every refined instance for that family and scope, followed by the mean runtime ratio at that cap. A dash means that no tested cap succeeds.}
  \centering
  \small
  \PaperTableRowsSetup
\begin{tabularx}{\linewidth}{>{\raggedright\arraybackslash}Xrrrr}
\PaperTableHideRows
\toprule
Graph family & Scope size & Mean width ratio & \multicolumn{2}{c}{Exact solver at minimum successful cap} \\
\cmidrule(l){4-5}
 & & & Width cap & Mean runtime ratio \\
\midrule
\PaperTableBodyRows
Banded & 2 & 2.50 & 6 & 67.8$\times$ \\
Banded & 3 & 4.00 & 8 & 3992.0$\times$ \\
Banded & 4 & 5.50 & -- & -- \\
\addlinespace[2pt]
Clustered & 2 & 2.50 & 6 & 60.9$\times$ \\
Clustered & 3 & 4.00 & 8 & 3539.8$\times$ \\
Clustered & 4 & 5.50 & -- & -- \\
\addlinespace[2pt]
Tree & 2 & 3.00 & 4 & 9.6$\times$ \\
Tree & 3 & 5.00 & 6 & 121.7$\times$ \\
Tree & 4 & 7.00 & 8 & 1453.6$\times$ \\
\hiderowcolors
\bottomrule
\end{tabularx}

\end{table}

Supplementary Table~\ref{SItab:formal_refinement_stress} quantifies Corollary~\ref{SIcor:treewidth_refinement} and Remark~\ref{SIrem:treewidth_refinement_tight}.

\paragraph{Search reliability}
The stability-screen transfer uses four workloads, 20 count resamples, 10 solver seeds, and a reduced schedule of four restarts, five $\alpha$ values, and 250 iterations, compared with a $12\times9\times800$ reference schedule.
We flag the reduced-schedule baseline count when its support falls below a chosen threshold under either count resampling or solver-seed variation.
The returned candidate set contains this baseline count and any count selected in at least 10\% of either perturbation family.
\begin{table}[tbp]
  \caption{\label{SItab:stability_guardrail_transfer}
  Stability-screen threshold ablation.
  At the retained $80\%$ threshold, the screen catches all \GuardrailDisagreements{} schedule disagreements, and the returned sets contain the reference count in \GuardrailCandidateRecoveries{} of those cases.
  Entries report flagged cases out of all cases, caught disagreements out of all disagreements, remaining disagreements out of unflagged cases, and disagreements for which the returned candidate set contains the reference domain count.
  Support is measured across count resamples and solver-seed perturbations. The screen is a review rule, not a correctness guarantee.}
  \centering
  \small
  \PaperTableRows
  \PaperTableRowsSetup
\begin{tabular}{ccccc}
\PaperTableHideRows
\toprule
\shortstack{Support\\threshold} & \shortstack{Flagged\\cases} & \shortstack{Caught\\mismatches} & \shortstack{Unflagged\\mismatches} & \shortstack{Reference-$k$\\coverage} \\
\midrule
\PaperTableBodyRows
60\% & 8/32 & 5/6 & 1/24 & 5/6 \\
70\% & 12/32 & 6/6 & 0/20 & 5/6 \\
80\% & 13/32 & 6/6 & 0/19 & 5/6 \\
90\% & 20/32 & 6/6 & 0/12 & 5/6 \\
\hiderowcolors
\bottomrule
\end{tabular}

\end{table}

\FloatBarrier

\subsection{Q2: Workload, risk, and crossing-cost effects}

\paragraph{Train-Ticket instantiation}\label{SIsup:train_ticket_instantiation}
The workload uses \TrainTicketTraces{} unique public Jaeger traces from seven Train-Ticket configurations.\cite{Steidl2022TrainTicketAnomalies}
Cross-service \texttt{CHILD\_OF} relations produce a directed parent-to-child call graph. Retaining the largest weakly connected component leaves \TrainTicketServices{} services and \TrainTicketEdges{} edges.
Aggregated edge counts are rescaled to the stated calls per top-level request.
Let $m(v)$ be the sum of incoming and outgoing trace-call multiplicities incident to service $v$, and let $m_{\max}=\max_v m(v)$.
The planning scenarios set
$w(v)=1.4+1.6\log(1+m(v))/\log(1+m_{\max})$.
Let $h(v)\in[0,1]$ be the first 32 bits of the SHA-256 hash of the service name scaled by $2^{32}-1$, let $o(v)$ be its outgoing call count, and let $o_{\max}=\max_v o(v)$.
Let $a(v)$ indicate that the name contains \texttt{admin}, and let $q(v)$ indicate an auth, user, payment, security, or assurance term.
\begin{samepage}
The baseline compromise prior is the scenario construction
\[
\operatorname{clip}_{[0.005,0.08]}\!\left(
0.015+0.020h(v)+0.006a(v)+0.005q(v)
+0.004\frac{\log(1+o(v))}{\log(1+o_{\max})}
\right).
\]
The reported $w(v)$ and $p(v)$ values are rounded to three decimal places and define a reproducible heterogeneous planning scenario. No empirical service-compromise frequencies are used.
\end{samepage}
Role-derived cluster labels affect only the skewed-risk multiplier.

\paragraph{Recommendation stability}\label{SIsup:recommendation_stability}
Recommendation stability uses 100 multinomial edge-count resamples with fixed solver seeds, separated from 20 solver seeds on the observed edge counts.
The shared reference-run value $k=\TrainDomains$ is recovered in at least $\TrainCountMatch\%$ of count resamples and $\TrainSolverMatch\%$ of solver seeds across the two risk scenarios at $B=0.10$\,ms.
The reported $k=6$ is right-censored at the tested upper limit.
The count experiment perturbs aggregate edge frequencies and does not model within-trace dependence.

For the critical-path sensitivity check, we extract the longest-duration root-to-leaf span chain from each of the \TrainTicketTraces{} traces.
For edge $(u,v)$, let $s_{uv}$ be the fraction of its observed calls appearing on those chains. Its multiplier is $(0.5+3s_{uv})$ normalized to call-weighted mean one.
Reweighting leaves the selected domain count unchanged in all $\CriticalStableRows$ risk--budget rows.
At $B=0.10$\,ms, it raises $\mathrm{BR}/\mathrm{BR}_0$ from $\CriticalHUniform$ to $\CriticalHWeighted$ under the heterogeneous baseline.
Under cluster skew, the ratio rises from $\CriticalSUniform$ to $\CriticalSWeighted$. At $B=0.20$ the change is at most $\CriticalLooseMaxChange$.

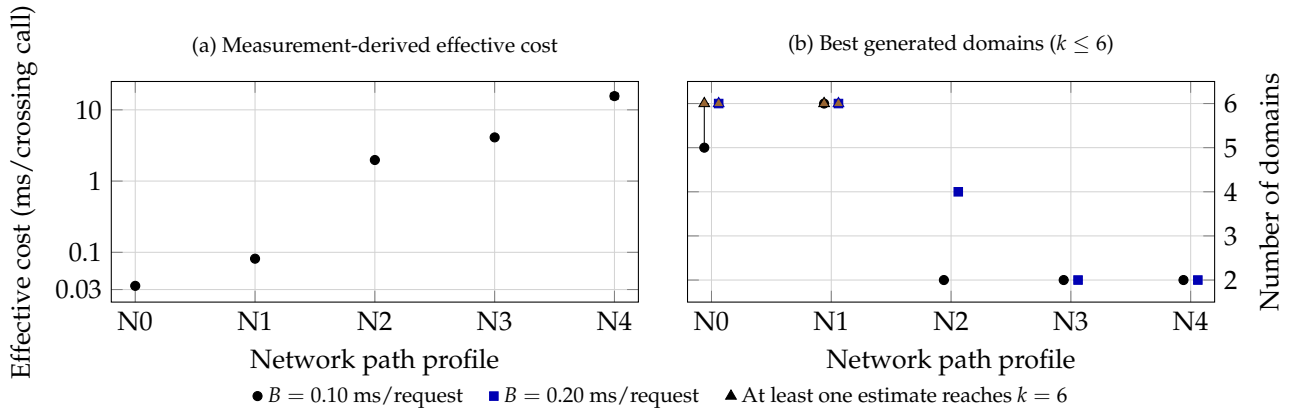
\begin{figure}[!ht]
  \centering
  \begin{tikzpicture}
\begin{axis}[
  name=pqcpathcost,
  width=0.47\textwidth,
  height=4.5cm,
  xmin=-0.2,
  xmax=4.2,
  xtick={0,1,2,3,4},
  xticklabels={N0,N1,N2,N3,N4},
  xlabel={Network path profile},
  ymode=log,
  ymin=0.02,
  ymax=25,
  ytick={0.03,0.1,1,10},
  yticklabels={0.03,0.1,1,10},
  title={(a) Measurement-derived effective cost},
  title style={font=\footnotesize},
  ylabel={Effective cost (ms/crossing call)},
  ylabel style={align=center},
  grid=major,
  grid style={line width=.1pt,draw=gray!35},
]
\addplot+[
  only marks,
  mark=*,
  mark size=1.7pt,
  mark options={fill=black},
  color=black,
  error bars/.cd,
  y dir=both,
  y explicit,
]
  table[x=path_index,y=a1_cost_ms,y error minus=a1_cost_err_minus,y error plus=a1_cost_err_plus,col sep=comma]
  {data/fig_inputs/pqc_path_architecture_figure.csv};
\end{axis}

\begin{axis}[
  name=pqcpathdomains,
  at={(pqcpathcost.east)},
  anchor=west,
  xshift=0.65cm,
  width=0.47\textwidth,
  height=4.5cm,
  xmin=-0.2,
  xmax=4.2,
  xtick={0,1,2,3,4},
  xticklabels={N0,N1,N2,N3,N4},
  xlabel={Network path profile},
  ymin=1.5,
  ymax=6.5,
  ytick={2,3,4,5,6},
  title={(b) Best generated domains ($k\leq6$)},
  title style={font=\footnotesize},
  ylabel={Number of domains},
  yticklabel pos=right,
  ylabel style={at={(axis description cs:1.14,0.5)},anchor=south},
  grid=major,
  grid style={line width=.1pt,draw=gray!35},
  legend to name=pqcpathlegend,
  legend columns=3,
  legend style={draw=none,fill=none,font=\footnotesize,/tikz/every even column/.append style={column sep=0.8em}},
]
\addplot+[
  only marks,
  mark=*,
  mark size=1.7pt,
  mark options={fill=black},
  color=black,
  error bars/.cd,
  y dir=both,
  y explicit,
]
  table[x=x_b010,y=k_b010,y error minus=k_b010_err_minus,y error plus=k_b010_err_plus,col sep=comma]
  {data/fig_inputs/pqc_path_architecture_figure.csv};
\addlegendentry{$B=0.10$ ms/request}
\addplot+[
  only marks,
  mark=square*,
  mark size=1.6pt,
  mark options={fill=blue!70!black},
  color=blue!70!black,
  error bars/.cd,
  y dir=both,
  y explicit,
]
  table[x=x_b020,y=k_b020,y error minus=k_b020_err_minus,y error plus=k_b020_err_plus,col sep=comma]
  {data/fig_inputs/pqc_path_architecture_figure.csv};
\addlegendentry{$B=0.20$ ms/request}
\addplot+[only marks,mark=triangle*,mark size=2.3pt,color=black,forget plot,filter discard warning=false]
  table[x=x_b010,y=k_b010_cap,col sep=comma]
  {data/fig_inputs/pqc_path_architecture_figure.csv};
\addplot+[only marks,mark=triangle*,mark size=2.3pt,color=blue!70!black,forget plot,filter discard warning=false]
  table[x=x_b020,y=k_b020_cap,col sep=comma]
  {data/fig_inputs/pqc_path_architecture_figure.csv};
\addlegendimage{only marks,mark=triangle*,mark size=2.3pt,color=black}
\addlegendentry{At least one estimate reaches $k=6$}
\end{axis}
\node[anchor=north] at ($(pqcpathcost.south)!0.5!(pqcpathdomains.south)+(0,-0.85cm)$)
  {\pgfplotslegendfromname{pqcpathlegend}};
\end{tikzpicture}
  \caption{\label{SIfig:pqc_path_architecture}
  Higher measured crossing costs reduce the number of domains in the best generated Train-Ticket designs.
  Panel (a) shows the effective mutual-authentication cost of X25519+ML-KEM-768 with ML-DSA-65 across five path profiles. Bars span stored bootstrap 5th to 95th percentile estimates.
  Panel (b) applies each estimate uniformly to all interaction edges under two boundary-latency budgets. Vertical bars span the selected domain counts across the three cost estimates.
  N0--N4 denote the path profiles whose round-trip time, rate, and loss appear in Table~\ref{SItab:pqc_named_profile}. $B$ is the latency budget in ms/request, and $k$ is the selected number of trust domains.
  Only N0 at $B=0.10$ changes across its interval, from five to six domains. Triangles mark the $k=6$ search cap when at least one of the three cost estimates selects six domains.}
\end{figure}

\paragraph{Named PQC calibration}
The measurement matrix contains \PQCMeasurementObservations{} observations from \PQCMeasurementBlocks{} randomized blocks per condition on an ARM64 host using OpenSSL~3.6.3.
C0 uses X25519 with ECDSA-P256 authentication, K1 uses hybrid X25519+ML-KEM-768 with ECDSA-P256 authentication, and A1 uses hybrid X25519+ML-KEM-768 with ML-DSA-65 authentication.
Each profile is measured with fresh and resumed sessions under server-only and mutual authentication across the five network paths in Table~\ref{SItab:pqc_named_profile}.
The displayed effective costs use mutual authentication, 20 requests per connection, and a $0.05$ full-handshake probability. The architecture replay separately assumes 20 calls per top-level request.

\begin{table}[!ht]
  \caption{\label{SItab:pqc_named_profile}
  For A1, the measured crossing-cost range across network paths exceeds the largest within-path spread among C0, K1, and A1.
  In the center-estimate A1 replay, the two lower-cost paths select more domains than the three higher-cost paths at both budgets.
  C0 uses X25519 with ECDSA-P256 authentication. K1 uses hybrid X25519+ML-KEM-768 key exchange with ECDSA-P256 authentication. A1 uses the same hybrid key exchange with ML-DSA-65 authentication.
  RTT is round-trip time, and a dash in the rate column denotes an unshaped path. $B$ is the latency budget in ms/request. Design cells report the conservative risk ratio $\mathrm{BR}_{\mathrm{node}}/\mathrm{BR}_0$, normalized to the single-domain baseline $\mathrm{BR}_0$, with the selected trust-domain count $k$ in parentheses.
  N0 is the unshaped local path. Across the stored bootstrap estimates, only N0 at $B=0.10$ changes, from five to six domains. A reported $k=6$ reaches the search cap.}
  \centering
  \small
  \setlength{\tabcolsep}{3pt}
  \PaperTableRowsSetup
  \begin{tabularx}{\linewidth}{l*{8}{>{\raggedleft\arraybackslash}X}}
    \PaperTableHideRows
    \toprule
    Path & \multicolumn{3}{c}{Network conditions} & \multicolumn{3}{c}{Crossing cost (ms/crossing call)} & \multicolumn{2}{c}{A1 design: $\mathrm{BR}/\mathrm{BR}_0$ ($k$)} \\
    \cmidrule(lr){2-4}\cmidrule(lr){5-7}\cmidrule(l){8-9}
    & RTT (ms) & Rate (Mbit/s) & Loss (\%) & C0 & K1 & A1 & $B=0.10$ & $B=0.20$ \\
    \midrule
    \PaperTableBodyRows
    N0 & 0 & -- & 0 & 0.0209 & 0.0194 & 0.0338 & 0.683 (5) & 0.619 (6) \\
N1 & 1 & 1000 & 0 & 0.0678 & 0.0726 & 0.0814 & 0.722 (6) & 0.662 (6) \\
N2 & 35 & 100 & 0.1 & 1.8809 & 1.8230 & 1.9756 & 0.944 (2) & 0.925 (4) \\
N3 & 70 & 20 & 1 & 3.9460 & 3.6150 & 4.1041 & 0.968 (2) & 0.944 (2) \\
N4 & 200 & 1 & 3 & 13.3346 & 14.5419 & 15.5815 & 0.968 (2) & 0.968 (2) \\
\hiderowcolors
\bottomrule

  \end{tabularx}
\end{table}

The architecture replay applies each profile/path coefficient uniformly to all interaction edges.
Absolute profile costs represent the full measured crossing cost. Incremental costs represent migration headroom relative to C0.
The profiles are deployment calibration inputs rather than comparative algorithm benchmarks. Their mean effective costs are not end-to-end tail-latency guarantees.

\begin{table}[!ht]
  \caption{\label{SItab:pqc_profile_architecture}
  Across the ten path--budget cells, center-estimate C0, K1, and A1 costs select the same domain count in eight cells. Incremental A1-minus-C0 costs select more domains than absolute A1 costs in six cells.
  The comparison uses Train-Ticket and the fixed bounded breadth-first derivation family.
  C0 uses X25519 with ECDSA-P256 authentication. K1 adds ML-KEM-768 to the key exchange, and A1 also uses ML-DSA-65 authentication. The $\Delta$A1 profile uses $c_{\mathrm{A1}}-c_{\mathrm{C0}}$ as the per-crossing migration cost.
  N0--N4 denote the path profiles listed by round-trip time, rate, and loss in Table~\ref{SItab:pqc_named_profile}. $B$ is the latency budget in ms/request.
  For each profile, $k$ is the selected trust-domain count. Risk ratio is $\mathrm{BR}_{\mathrm{node}}/\mathrm{BR}_0$, normalized to the single-domain baseline $\mathrm{BR}_0$. Lower ratios indicate less modeled impact. All profiles use the same optimizer schedule. A reported $k=6$ reaches the search cap.}
  \centering
  \small
  \setlength{\tabcolsep}{4pt}
  \PaperTableRows
  \begin{tabularx}{\linewidth}{lc*{4}{>{\centering\arraybackslash}X}*{4}{>{\raggedleft\arraybackslash}X}}
    \PaperTableHideRows
    \toprule
    Path & $B$ (ms/request) & \multicolumn{4}{c}{Domain count $k$} & \multicolumn{4}{c}{Risk ratio} \\
    \cmidrule(lr){3-6}\cmidrule(l){7-10}
    & & C0 & K1 & A1 & $\Delta$A1 & C0 & K1 & A1 & $\Delta$A1 \\
    \midrule
    \PaperTableBodyRows
    N0 & 0.10 & 6 & 6 & 5 & 6 & 0.657 & 0.630 & 0.683 & 0.614 \\
 & 0.20 & 6 & 6 & 6 & 6 & 0.610 & 0.607 & 0.619 & 0.608 \\
N1 & 0.10 & 6 & 6 & 6 & 6 & 0.704 & 0.704 & 0.722 & 0.613 \\
 & 0.20 & 6 & 6 & 6 & 6 & 0.662 & 0.673 & 0.662 & 0.606 \\
N2 & 0.10 & 2 & 2 & 2 & 6 & 0.944 & 0.944 & 0.944 & 0.722 \\
 & 0.20 & 3 & 3 & 4 & 6 & 0.904 & 0.904 & 0.925 & 0.687 \\
N3 & 0.10 & 2 & 2 & 2 & 5 & 0.968 & 0.968 & 0.968 & 0.783 \\
 & 0.20 & 2 & 2 & 2 & 6 & 0.944 & 0.944 & 0.944 & 0.722 \\
N4 & 0.10 & 2 & 2 & 2 & 2 & 0.968 & 0.968 & 0.968 & 0.944 \\
 & 0.20 & 2 & 2 & 2 & 4 & 0.968 & 0.968 & 0.968 & 0.925 \\
\hiderowcolors
\bottomrule

  \end{tabularx}
\end{table}

\FloatBarrier
\paragraph{Held-out graph replay}
For A1, five independent N0 runs and five independent N2 runs measure every Train-Ticket edge and session mode in 20 calibration trials, optimize at $B\in\{0.10,0.20\}$\,ms, and then remeasure only the selected crossing edges in 10 held-out trials.
Budget compliance is determined from the separately observed boundary latency.

Without a reserve, the edge-specific block maxima pass every N0 check and \PQCNTwoEdgePasses{} of the \PQCChecksPerPath{} N2 checks. The graph block maximum passes \PQCNZeroGraphPasses{} of the \PQCChecksPerPath{} N0 checks and every N2 check.
With a $5\%$ budget reserve, each rule passes all \PQCHeadroomChecksPerRule{} path--budget checks.
These checks support the reserve for the tested paths and budgets. Deployment therefore remains iterative: calibrate, optimize, remeasure the selected crossings, reject violations, and update costs or reserve before rerunning.

\FloatBarrier

\subsection{Q3: Explicit shared-issuer rescoring}
In the induced Train-Ticket cases, every generated candidate is evaluated with explicit issuer reachability to provide the reference. The additive score and a reachability-aware proxy are then compared as candidate-ranking rules.
The $J\!\to\!p_{\mathrm{eff}}$ reduction is exact under the antichain condition. When accepted services have overlapping descendant sets, the additive score remains an upper bound and final selection uses explicit evaluation.
Across the induced comparisons, the additive score identifies the best explicitly evaluated candidate in \OverlayCompiledRecoveries{} of \OverlayComparisons{} cases and misses by at most \OverlayMaxGap\%.
Evaluating the three candidates ranked highest by the reachability-aware proxy recovers the best generated candidate in all \OverlayComparisons{} cases.
The reported gaps are relative to the explicit best among generated candidates under the same budget. Global optimality lies outside this comparison.

\FloatBarrier

\section{Deployment considerations}\label{SIsup:deployment}
The model isolates \emph{credential authority} (who can mint or derive which identities) as the mechanism that determines blast-radius propagation under compromise.
The deployment class treats each trust domain as a key-establishment domain rooted in a common issuer or key-management authority.
After establishment, intra-domain payload interactions use symmetric protection or already-issued delegated credentials. Communication between independently rooted domains invokes public-key authentication or key establishment.
Deployments can add structure around that core through scoped credentials (role-based access control (RBAC) or attribute-based access control (ABAC)), policy-gated reachability, and finite compromise windows due to detection and rotation.
These effects can be represented as refinements without changing the cut-and-reachability structure.

\subsection{Network segmentation vs.\ trust domains}
A recurring deployment mistake is to equate network segmentation with cryptographic trust segmentation.
Network segmentation and trust domains describe different objects.
The model keeps them separate:
\begin{itemize}
  \item the interaction graph $G$ records which services communicate and therefore where latency-relevant crossings can occur,
  \item the domain assignment $D$ records which services share a cryptographic trust boundary, and
  \item the derivation graphs $H_i$ record how compromise propagates once an issuer, root, or delegated credential source is lost.
\end{itemize}
This distinction matters even on a tiny topology such as
\[
\texttt{client}\rightarrow\texttt{gateway}\rightarrow\texttt{orders}\rightarrow\texttt{payments}\rightarrow\texttt{ledger},
\]
with an auxiliary issuer \texttt{auth}.
Several readings of the same service graph are possible:
\begin{itemize}
  \item \emph{Gateway TLS only / collapsed internal segment.} TLS terminates at the gateway and internal services share one credential root or one implicit trusted segment.
  This yields low boundary latency. The internal segment behaves like one collapsed trust domain.
  \item \emph{Gateway plus isolated issuer.} The runtime topology is unchanged. Placing \texttt{auth} in its own trust domain prevents compromise of an internal service from automatically granting the issuer's authority.
  This adds a boundary crossing and can reduce blast radius.
  \item \emph{Per-service identities / service-to-service mutual TLS (mTLS).} Distinct service identities under one accepted trust bundle do not by themselves create distinct trust domains: compromise of their common authority retains shared minting power.
  Per-service trust domains require independently rooted credential systems, with cross-root authentication or key establishment on service interactions. This can reduce root-compromise reach while increasing boundary latency.
\end{itemize}
A VLAN, subnet, or gateway boundary should not be collapsed into one black-box node when the analysis target is credential blast radius.
Even if network topology is unchanged, moving a trust root or changing which services share an issuer can change propagation sharply.
JSON Web Token (JWT) issuer systems make this distinction explicit. The service-interaction graph may remain unchanged. Acceptance of one issuer by multiple services makes that issuer a shared propagation root regardless of network segmentation.
\MainSecRef{sec:explicit_prop_setting}{The explicit-propagation setting} allows additive issuer scoring when, for each issuer and domain, the accepted targets have disjoint descendant sets. In a chain this permits at most one target per domain. Direct issuance permits multiple targets because their descendant sets are disjoint.

\subsection{Extracting derivation structure from control-plane data}\label{SIsup:extract_hi}
For deployment, reconstruct $H_i$ from identity and key-management control planes rather than from network traces alone.
Audited control planes can expose the dominant trust roots and delegation layers.
Useful sources include service meshes, workload-identity systems, JWT/OIDC issuer configuration, and managed certificate authority (CA) or key management system (KMS) platforms.
In legacy or mixed-vendor environments, an exact reconstruction may be unavailable. Evaluate a small set of plausible $H_i$ candidates instead.

These control-plane sources map to the model as follows:
\begin{itemize}
  \item \emph{SPIFFE/SPIRE, service-mesh CA, workload identity.} Determine which workloads receive identities from each trust bundle or issuer and whether namespace- or cluster-level subissuers exist. This fixes roots, initial domain membership, and root-to-workload issuance edges. A star is appropriate when workloads receive credentials directly from a root. A depth-two approximation is appropriate when the control plane contains one intermediate issuer layer.
  \item \emph{JWT/OIDC issuer and verifier configuration.} Use issuer, audience, and scope configuration to determine which services accept each issuer and whether a token-minting service is a shared propagation root. For each issuer and domain, check whether any accepted target is an ancestor of another. If so, use explicit issuer reachability to avoid counting overlapping descendants more than once.
  \item \emph{CA/KMS and hardware security module (HSM) inventory and signing-service metadata.} Use key-custody and signing metadata to identify who signs for whom and whether intermediate or tenant issuers exist. This determines derivation edges and whether a star, chain, or depth-two abstraction is appropriate.
  \item \emph{Authorization and mesh policy.} Remove interactions that cannot occur from $E$ and exclude unavailable derivation arcs from $\mathcal A^{\mathrm{allow}}$. Check whether feasible derivation trees remain.
  \item \emph{RBAC and administrative scopes.} Use privileged-role inventories to decide which services should be split into scoped principals with separate $w$ and $p$ values.
\end{itemize}

The extraction sequence is to inventory dominant roots, map issuance or delegation edges, intersect them with verifier-acceptance and policy data, and collapse the result to the coarsest defensible planning abstraction before optimization.
Audited control planes support a direct reconstruction. Legacy systems require sensitivity analysis across plausible $H_i$ and $J$ families.

\subsection{Practitioner instantiation checklist}
This checklist maps common observability and security tooling to model inputs. Deployment-specific estimation is still required.
\begin{enumerate}
  \item \emph{Services/principals ($V$).} Decide the granularity: services only, or also explicit issuers/verifiers, sidecars, and gateways if they hold independent credentials.
  \item \emph{Interaction edges and rates ($E,r_{uv}$).} Use distributed tracing, such as OpenTelemetry or Jaeger, to extract a directed call graph and per-edge call counts or rates. When interpreting $\mathrm{Lat}(D)$ as milliseconds per request, normalize these values to calls per top-level request. If $c_{\mathrm{pqc}}$ is a raw per-event cost, convert the call rates to event-equivalent rates using the observed reuse and reauthentication policy. Do not apply this conversion when the cost is already amortized per interaction. Request classes or critical-path labels may provide additional edge weights.
  \item \emph{Latency sensitivity ($\ell_{uv}$).} Start with $\ell_{uv}=1$. Refine using path criticality (edges on paths critical to a service-level objective (SLO)), slack within traced request dependency graphs, or downstream fanout/queuing sensitivity when available.
  \item \emph{Crossing cost ($c_{\mathrm{pqc}}(u,v)$).} Microbenchmark the deployed cross-root key-establishment or authentication stack on representative hardware. Represent protocol amortization through either event-equivalent rates or effective per-interaction costs, not both.

  To model session state, estimate a cold/warm mixture $c_{\mathrm{eff}}(u,v)=c^{\mathrm{warm}}_{uv}+P^{\mathrm{cold}}_{uv}c^{\mathrm{cold}}_{uv}$.
  Here $c^{\mathrm{warm}}_{uv}$ is the reused-session cost, $c^{\mathrm{cold}}_{uv}$ is the additional setup penalty, and $P^{\mathrm{cold}}_{uv}$ is inferred from telemetry. Relevant factors include connection lifetime, request limits, stream limits, idle timeouts, and burst fanout.
  \item \emph{Criticality weights ($w(v)$).} Choose weights from business or mission impact or data sensitivity (e.g., assign greater weight to services handling personally identifiable information (PII) or payments). Normalize so $\sum_v w(v)$ is interpretable.
  \item \emph{Compromise probabilities ($p_T(x)$).} Pick an explicit planning horizon $T$ (e.g., expected detection/rotation window) and set $p_T(x)$ via exposure tiers or a hazard model. Represent uncertainty with intervals and design against upper endpoints.
  \item \emph{Policy gating.} If service-mesh authorization (mTLS + policy) or network reachability forbids some interactions, gate them by removing edges from $E$ (or setting $r_{uv}\ell_{uv}=0$). Similarly, remove derivation edges that cannot be exercised operationally.
  \item \emph{Scoped credentials (vertex refinement).} If a service holds multiple identities (e.g., normal service vs.\ admin/minting), split $v$ into scoped principals $(v,s)$, allocate weights/priors per scope, and distribute each base edge's traffic across scoped edges according to which scopes are exercised.
  \item \emph{Derivation constraints ($\mathcal A^{\mathrm{allow}},\Delta,h$).} Extract eligible derivation arcs from issuer policy and credential-control-plane metadata. Set fanout and depth from the intended key-management mechanism. Treat the remaining coupled optimization as a design-space search under these constraints.
\end{enumerate}

The same graph optimizer can be retained by refining the linear edge coefficients by request class:
\begin{equation}
  \mathrm{Lat}_{\mathrm{mean}}(D)
  \;=\;
  \sum_{(u,v)\in E} \mathbf{1}[D(u)\ne D(v)]
  \sum_{q\in\mathcal{Q}} \omega_q\,n_{uv,q}\,\bigl(\tau_{uv,q}+P^{\mathrm{cold}}_{uv,q}\kappa_{uv,q}\bigr).
\end{equation}
Here $n_{uv,q}$ is traced use of edge $(u,v)$ under request class $q$, $\omega_q$ is a class weight, $\tau_{uv,q}$ is the warm crossing cost, and $\kappa_{uv,q}$ is the additional cold-setup penalty.
These quantities are combined into weighted graph edges instead of being optimized over as a separate request-level directed acyclic graph (DAG) model.

Handshake flights, certificate growth, session reuse, retry behavior, and path limits enter through the calibrated edge coefficients and cold/warm mixture above.
Explicit protocol state machines, packet failure probabilities, and request-level precedence are outside the graph formulation.

\subsection{Operational refinements}
\paragraph{Scoped principals}
Services may hold credentials with different scopes, use rates, and risk.
Represent them by replacing each base vertex $v\in V$ with scoped principals $S_v$ and forming $V'=\{(v,s):v\in V,\,s\in S_v\}$.
Assign $w'(v,s)$ and $p'(v,s)$ per scope, then distribute each base edge's traffic among scoped edges while preserving its total weight.

Latency still depends only on whether an edge crosses domains, so the same procedures apply to $(V',E')$.
Refinement into at most $s$ principals per service increases the vertex count by at most a factor $s$. Under the direct-issuance assumptions of Proposition~\ref{SIprop:star_treewidth}, the bag-expansion argument preserves polynomial-time exact optimization for bounded width and constant $s$, with specified anchors or fixed $k$.
Refinement can isolate highly privileged principals behind stricter boundaries or separate issuers at low latency cost when their interaction rate is small.

\paragraph{Policy gating and reachability}
Policy and network reachability determine which interactions can occur.
Remove disallowed interactions from $E$, set $r_{uv}\ell_{uv}=0$, or replace $G$ with an effective interaction graph $G_{\mathrm{eff}}$ derived from reachability and authorization policy.
If compromise of $x$ cannot exercise a derivation edge, exclude that edge from $\mathcal A^{\mathrm{allow}}$. Reconstructed or redesigned trees must still satisfy (C3)--(C4).
The resulting model remains a cut-and-reachability problem with \emph{policy-conditioned} graphs.

\paragraph{Finite compromise windows}
Our compromise probabilities $p(x)$ are marginals over a planning horizon.
To make that horizon explicit, one can parameterize $p_T(x)$ by an exposure window $T$ (e.g., mean time to detect/rotate), using a simple hazard model such as $p_T(x)=1-\exp(-\lambda_x T)$.
Uncertainty in $\lambda_x$ or $T$ can be handled with interval bounds and robust design as in Supplementary Corollary~\ref{SIcor:interval_robustness}, yielding worst-case containment guarantees over plausible deployment regimes.

\end{document}